\documentclass[english]{article}
\usepackage{lmodern}

\usepackage[T1]{fontenc}
\usepackage[latin9]{inputenc}
\usepackage[a4paper]{geometry}
\usepackage{color}
\usepackage{babel}
\usepackage{float}
\usepackage{amsmath}
\usepackage{amsthm}
\usepackage{amssymb}
\usepackage{stmaryrd}
\usepackage{graphicx}
\usepackage[authoryear]{natbib}
\usepackage{appendix}
\usepackage[unicode=true,
 bookmarks=true,bookmarksnumbered=false,bookmarksopen=false,
 breaklinks=true,pdfborder={0 0 0},pdfborderstyle={},colorlinks=true]
 {hyperref}
\hypersetup{pdftitle={A Gibbs sampler for a class of random convex polytopes},
 pdfauthor={PEJ, RG, PTE, APD}}
 \usepackage{authblk}

\makeatletter

\floatstyle{ruled}
\newfloat{algorithm}{tbp}{loa}
\providecommand{\algorithmname}{Algorithm}
\floatname{algorithm}{\protect\algorithmname}

\numberwithin{equation}{section}
\theoremstyle{plain}
\newtheorem{thm}{\protect\theoremname}[section]
\theoremstyle{plain}
\newtheorem{lem}[thm]{\protect\lemmaname}
\theoremstyle{definition}

\theoremstyle{remark}

\theoremstyle{plain}
\newtheorem{prop}[thm]{\protect\propositionname}

\newcommand{\pp}{\texttt p}
\newcommand{\qq}{\texttt q}
\newcommand{\rr}{\texttt r}

\usepackage{dsfont}
\usepackage[dvipsnames,svgnames,x11names,hyperref]{xcolor}
\hypersetup{urlcolor=RoyalBlue,citecolor=RoyalBlue}
\usepackage[normalem]{ulem}

\@ifundefined{showcaptionsetup}{}{%
 \PassOptionsToPackage{caption=false}{subfig}}
\usepackage{subfig}
\makeatother

\providecommand{\examplename}{Example}
\providecommand{\lemmaname}{Lemma}
\providecommand{\propositionname}{Proposition}
\providecommand{\remarkname}{Remark}
\providecommand{\theoremname}{Theorem}

\SetSymbolFont{stmry}{bold}{U}{stmry}{m}{n}

\title{A Gibbs sampler for a class of random convex polytopes}
\author[1]{Pierre E. Jacob}
\author[2]{Ruobin Gong}
\author[3]{Paul T. Edlefsen}
\author[1]{Arthur P. Dempster}
\affil[1]{Department of Statistics, Harvard University}
\affil[2]{Department of Statistics, Rutgers University}
\affil[3]{Fred Hutchinson Cancer Research Center}

\usepackage{setspace}
\begin{document}
\maketitle

\begin{abstract}
We present a Gibbs sampler for the Dempster--Shafer (DS) approach to
statistical inference for Categorical distributions. The DS framework extends the
Bayesian approach, allows in particular the use of partial prior information,
and yields three-valued uncertainty assessments representing
probabilities ``for'', ``against'', and ``don't know'' about formal assertions
of interest.   The
proposed algorithm targets the distribution of a class of random convex
polytopes which encapsulate the DS inference. The sampler relies on an
equivalence between the iterative constraints of the vertex configuration and
the non-negativity of cycles in a fully connected directed graph.
Illustrations include the testing of independence in $2\times 2$ contingency
tables and parameter estimation of the linkage model.
\end{abstract}

\section{Introduction \label{sec:intro}}

Consider observed counts of $K$ possible categories, denoted by
$N_1,\ldots,N_K$ and summing to $N$.  We assume that these counts are 
sums of independent draws from a Categorical distribution. The goal is to infer the associated
parameters $\theta$ in the simplex of dimension $K$ and to forecast future
observations. The setting is most familiar to statisticians and
if $K$ is small relative to $N$, and without further information about
$\theta$, the story is somewhat simple with the maximum likelihood estimator
being both very intuitive and efficient. 
The plot thickens quickly if $N$ is small or indeterminate, if
partial prior information is available, if observations are imperfect, or if
additional constraints are imposed, especially when uncertainty
quantification is simultaneously sought 
\citep{fitzpatrick1987quick,berger1992ordered,sison1995simultaneous,liu2000estimation,lang2004multinomial,chafai2009confidence,dunson2009nonparametric}.
As any probability distribution on a finite unordered set is necessarily Categorical,  the
setting often arises as part of 
more elaborate procedures.
Besides, the canonical nature of Categorical distributions has made them a common test bed for
various approaches to inference
\citep{walley1996inferences,bernard1998bayesian}. 

The Dempster--Shafer (DS) theory is a framework for probabilistic reasoning
based on observed data and modeling of knowledge.  In the DS framework,
inferences on user-defined assertions are expressed probabilistically. These
assertions can be statements concerning parameters (``the parameter belongs to
a certain set'') or concerning future observations. Contrary to Bayesian
inference, no prior distribution is strictly required, and partial prior
specification is allowed (see Section \ref{subsec:partialprior}).  
Rather than posterior probabilities, DS inference
yields three-valued assessments of uncertainty, namely probabilities ``for'',
``against'', and ``don't know'' associated with the assertion of interest, and
denoted by $(\pp, \qq, \rr)$ (see Section \ref{subsec:dsinference}).  In his pioneering work,
\cite{dempster1963direct,Dempster66,dempster1967upper,
dempster1968generalization}
developed the idea of upper and lower probabilities for 
assertions of interest. Together with
the contributions of 
\cite{shafer1976mathematical,shafer1979allocations}, the approach 
became known as the Dempster--Shafer (DS) theory of belief functions.
The framework has various connections to other ways of obtaining lower and
upper probabilities and to robust Bayesian inference
\citep{wasserman1990belief}. Over the past decades, the DS theory saw various
applications in signal processing, computer vision and machine learning \citep[see
e.g.,][]{bloch1996some,vasseur1999perceptual,denoeux2000neural,basir2007engine,denoeux2008k,diaz2010shape}.
As outlined in
\cite{Dempster:IJAR,dempster2014statistical} the DS 
framework is an ambitious tool for carrying out statistical inferences in
scientific practice.  During the developments of
DS, Categorical distributions were front and center due to their generality
and relevance to ubiquitous statistical objects such as contingency tables. 

The computation required by the DS approach for Categorical distributions
proved to be demanding. The approach
involves distributions of convex polytopes within the simplex, some properties
of which were found in
\citet{Dempster66,dempster1968generalization,dempster1972class}.
Unfortunately, no closed-form joint distribution of the vertices has 
been found,
hindering both theoretical developments and numerical 
approximations. The challenge prompted
\citet{denoeux2006constructing} to comment that, ``Dempster studied the
trinomial case {[}...{]} However, the application of these results to compute
the marginal belief function {[}\dots{]} has proved, so far and to our
knowledge, mathematically intractable.'' Likewise, \citet{lawrence2009new}
commented: ``{[}...{]} his method for the Multinomial model is seldom used,
partly because of its computational complexity.'' Over the past fifty years,
the literature saw a handful of alternative methods for Categorical inference
via generalized fiducial inference \citep{Hannig2016JASA, liu2016generalized}
the Imprecise
Dirichlet Model \citep{walley1996inferences}, the Dirichlet-DSM method
\citep{lawrence2009new}, the vector-valued Poisson model
\citep{edlefsen2009estimating}, and the Poisson-projection method for
Multinomial data \cite[unpublished PhD thesis]{gong2018thesis}. 
The latter three methods were motivated in part to circumvent the
computational hurdle put forward by the original DS formulation.
The present article aims at filling that gap by proposing an
algorithm that carries out the computation proposed in 
\citet{Dempster66,dempster1972class}. The
presentation does not assume previous 
knowledge on DS inference.

Section~\ref{sec:categorical} introduces the formal setup.
Section~\ref{sec:gibbssampler} presents an
equivalence between constraints arising in the definition of the 
problem and the existence of ``negative cycles'' (defined in Section \ref{subsec:constraints}) in a certain weighted graph,
that leads to a Gibbs sampler.  Various
illustrations and extensions are laid out in Section~\ref{sec:manipulations}. 
Section~\ref{sec:application} concerns applications 
to $2\times 2$ contingency tables and 
the linkage model.  Elements of future research are discussed 
in Section \ref{sec:discussion}.  
Code in R \citep{Rsoftware} is available at
\href{https://github.com/pierrejacob/dempsterpolytope}{github.com/pierrejacob/dempsterpolytope}
to reproduce the figures of the article.

\section{Inference in Categorical distributions \label{sec:categorical}}

We describe inference in Categorical distributions
as proposed in \citet{Dempster66}, using the following notation.
The observations are $\mathbf{x} =
(x_{n})_{n\in[N]}$, with $x_n \in [K]$ for all $n\in[N]$, where $[m]$ denotes
the set $\{1,\ldots,m\}$ for $m\geq1$. The number of categories is $K$.  The
$K$-dimensional simplex is
$\Delta=\{z\in\mathbb{R}_{+}^{K}\colon\;\sum_{k=1}^{K}z_{k}=1\}$.  The set of
measurable subsets of $\Delta$ is denoted by $\mathcal{B}(\Delta)$.  We denote
the vertices of $\Delta$ by $V_{1},\ldots,V_{K}$. In barycentric coordinates,
$V_{k}$ is a $K$-vector with $k$-th entry equal to one and other entries equal
to zero.  A polytope is a set of points $z\in\mathbb{R}^{K}$ satisfying linear
inequalities, of the form $Mz\leq c$ understood component-wise, and where $M$
is a matrix with $K$ columns and $c$ is a vector.  For a given
$\mathbf{x}\in[K]^N$, $\mathcal{I}_{k}$ is the set of indices $\{n\in[N]\colon
x_n = k\}$.  The counts are $N_k = |\mathcal{I}_k|$
and $\sum_{k\in [K]} N_k = N$.  Coordinates of $u_{n}\in\Delta$ are denoted
by $u_{n,k}$ for $k\in[K]$.  The volume of a set $A$ is denoted by
$\text{Vol}(A)$. The uniform variable $Z$ over $\mathcal{S}$ is written
$Z\sim \mathcal{S}$.

\subsection{Sampling mechanism and feasible sets\label{subsec:samplingmechanism}}

The goal is to infer the parameters $\theta=(\theta_{1},\ldots,\theta_{K})\in\Delta$
of a Categorical distribution using observation $\mathbf{x} = (x_{n})_{n\in[N]}\in[K]^N$.
Viewing $x_n$ as a random variable, the model states $\mathbb{P}(x_{n}=k)=\theta_{k}$ for all $k\in[K]$,
$n\in[N]$. Generating draws from a Categorical distribution can be done
in different ways.  
In the DS approach, the choice of sampling mechanism of the observable data has
an impact on inference of the parameters, a feature that distinguishes DS from
likelihood-based approaches, and aligns it with fiducial
\citep{fisher1935fiducial,Hannig2016JASA}, structural
\citep{fraser1968structure}, and functional \citep{dawid1982functional}
approaches. 
Appendix \ref{appx:sampling-mechanism} illustrates this impact in a simple setting.
We follow \citet{Dempster66} and consider the following sampling
mechanism for $x_n$, which is invariant by permutation of the labels of the categories;
it is equivalent to the ``Gumbel-max trick'' \citep{maddison2014sampling} as 
explained in Appendix \ref{appx:sampling-gumbelmax}.
Given $\theta$, 
for each $k\in[K]$, define $\Delta_{k}(\theta)$
to be a ``subsimplex'' obtained as the polytope with the same vertices as $\Delta$ except that
vertex $V_{k}$ is replaced by $\theta$.
The sets $(\Delta_{k}(\theta))_{k\in[K]}$ form a partition of $\Delta$, shown
in Figure \ref{fig:sdk:subsimplex}. It can be checked that $\text{Vol}(\Delta_k(\theta)) = \theta_k$.
Then, introduce $u_n\sim \Delta$, and 
define $x_n$ as 
\begin{equation}\label{eq:sampling-mechanism}
x_{n}=\sum_{k\in[K]}k\mathds{1}(u_{n}\in\Delta_{k}(\theta)).
\end{equation}
In other words, $x_n$ is the unique index $k\in[K]$ such that $u_{n}$ belongs to
$\Delta_{k}(\theta)$. Since $\text{Vol}(\Delta_k(\theta)) = \theta_k$, $x_n$ indeed follows the 
Categorical distribution with parameter $\theta$. 
Lemma \ref{Lemma-5.2} recalls a useful
characterization of $\Delta_{k}(\theta)$.
\begin{lem}
    \label{Lemma-5.2}(Lemma 5.2 in \citet{Dempster66}). For $k\in[K]$, $\theta\in\Delta$
    and $u_n\in \Delta$, $u_n\in \Delta_{k}(\theta)$ if and only if  $u_{n,\ell}/u_{n,k} \geq \theta_\ell / \theta_k$
for all $\ell\in[K]$.
\end{lem}

Given fixed observations $\mathbf{x} = (x_n)_{n\in[N]}$, the sampling mechanism in \eqref{eq:sampling-mechanism}
can be turned into constraints on the values of $\mathbf{u} = (u_n)_{n\in[N]}$ and $\theta$ that could
have led to the observations.
A central piece of the
machinery is the following set,
\begin{equation}
\mathcal{R}_{\mathbf{x}}=\left\{ (u_{1},\ldots,u_{N})\in\Delta^{N}\colon\;
\exists\theta\in\Delta\quad\forall n\in[N]\quad 
u_{n} \in \Delta_{x_n}(\theta) \right\} .\label{eq:RN}
\end{equation}
It is the set of all possible realizations of $\mathbf{u}$ which could have produced
the data $\mathbf{x}$ for (at least) some $\theta$, via the specified sampling mechanism. 
Given a realization of $\mathbf{u}\in\mathcal{R}_{\mathbf{x}}$ by definition
there is a non-empty ``feasible'' set $\mathcal{F}(\mathbf{u})\subset \Delta$ defined
as
\begin{equation}
    \mathcal{F}(\mathbf{u})=\left\{ \theta\in \Delta\colon\;\forall n\in[N]\quad u_{n} \in \Delta_{x_n}(\theta)\right\}.\label{eq:FA}
\end{equation}
On the other hand if $\mathbf{u}$ is an arbitrary point in $\Delta^N$ then 
$\mathcal{F}(\mathbf{u})$ defined above can be empty.
For example $\mathbf{u}=(u_1,u_2,u_3)$ shown in Figure \ref{fig:sdk:cups}
leads to an empty $\mathcal{F}(\mathbf{u})$ for the observations $x_1=1,x_2=3,x_3=2$.
The goal of the proposed method is to obtain 
non-empty sets $\mathcal{F}(\mathbf{u})$,
as illustrated for another data set in Figure
\ref{fig:gibbs:1}.
We can rewrite \eqref{eq:RN} as $\mathcal{R}_{\mathbf{x}} =
\{\mathbf{u}\colon \mathcal{F}(\mathbf{u})\neq \emptyset\}$. 

The ingredients introduced thus far specify the ``source'' of a belief function
\citep[e.g.][]{wasserman1990belief}. The central object of interest here is the
distribution of the random sets $\mathcal{F}(\mathbf{u})$
conditional on them being non-empty.
We consider the uniform distribution on $\mathcal{R}_{\mathbf{x}}$
denoted by $\nu_{\mathbf{x}}$, with density 
\begin{equation}
    \forall u_{1},\ldots,u_{N}\in \Delta^{N}\quad\nu_{\mathbf{x}}(u_{1},\ldots,u_{N})=\text{Vol}\left(\mathcal{R}_{\mathbf{x}}\right)^{-1}\mathds{1}\left((u_{1},\ldots,u_{N})\in\mathcal{R}_{\mathbf{x}}\right).\label{eq:nuNuniformonRN}
\end{equation}
Our main contribution is an algorithm 
to sample $\mathbf{u}$ from $\nu_{\mathbf{x}}$. The sets $\mathcal{F}(\mathbf{u})$
obtained when $\mathbf{u} \sim \nu_{\mathbf{x}}$ constitute the class of random convex polytopes
studied in \citet{dempster1972class} and referred to in the title of the present article.
The distribution 
$\nu_{\mathbf{x}}$ is also the result of Dempster's rule of combination 
\citep{dempster1967upper}
applied to the information provided separately by each of the $N$ observations.

\begin{figure}[t]
\begin{centering}
    \subfloat[\label{fig:sdk:subsimplex}]{\begin{centering}
\includegraphics[width=0.35\textwidth]{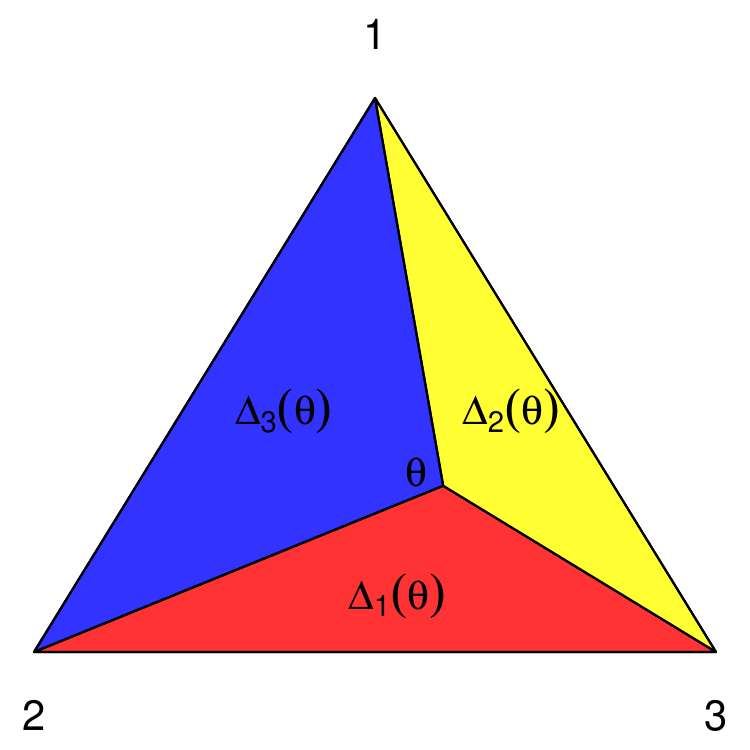}
\par\end{centering}
}
\hspace*{1cm}
\subfloat[\label{fig:sdk:cups}]{\begin{centering}
    \includegraphics[width=0.35\textwidth]{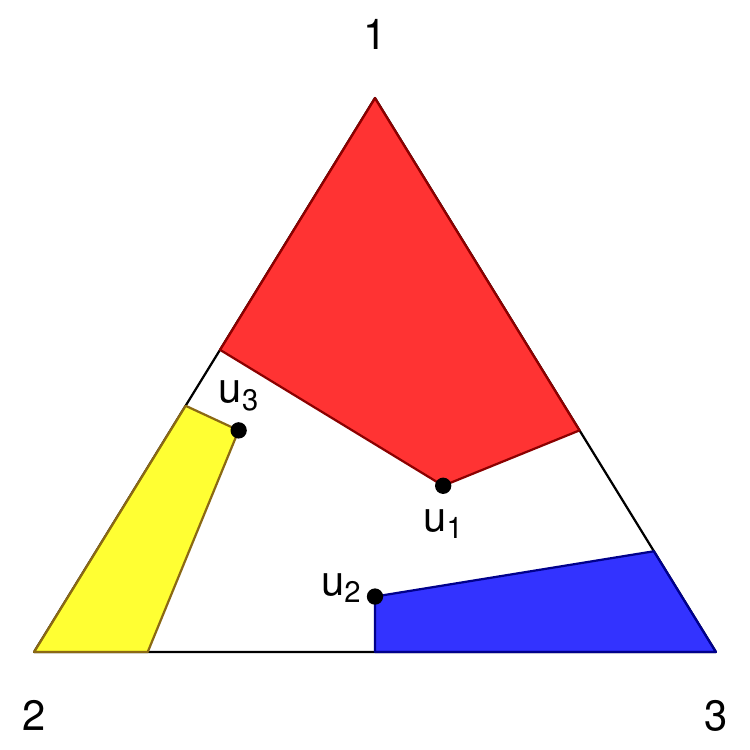}
\par\end{centering}
}
\par\end{centering}
\caption{Partition of $\Delta$ into $(\Delta_{k}(\theta))_{k\in[K]}$ in \ref{fig:sdk:subsimplex},
with $K=3$.
Each point $u_n\in\Delta$ defines,
for a fixed $x_n\in[K]$, a set of $\theta \in \Delta$ such that $u_n\in \Delta_{x_n}(\theta)$;
three such sets are colored in \ref{fig:sdk:cups}, for  $x_1=1,x_2=3,x_3=2$.
Here, no $\theta\in\Delta$ is such that $u_n\in \Delta_{x_n}(\theta)$ for $n=1,2,3$.
\label{fig:sdk:triangle}}
\end{figure}

\subsection{Inference using random sets \label{subsec:dsinference}}

We recall briefly how random sets can be processed into  ``lower'' and
``upper'' probabilities as in 
\citet{Dempster66}, or into ``belief'' and
``plausibility'' as in
\citet{shafer1976mathematical,shafer1990perspectives,wasserman1990belief}, or
$(\pp, \qq, \rr)$ probabilities as in \citet{Dempster:IJAR}.  The user provides
a measurable subset $\Sigma \in \mathcal{B}(\Delta)$ corresponding to an
``assertion'' of interest about the parameter, for instance $\Sigma=\{\theta\in \Delta \colon \theta_1 \leq 1/3\}$, or
$\Sigma=\{\theta\in\Delta\colon \theta_1/\theta_2 > \theta_3/\theta_4\}$.  The belief
function  assigns a value to each $\Sigma \in \mathcal{B}(\Delta)$ defined as
$\text{Bel}(\Sigma)=\nu_{\mathbf{x}}(\{\mathbf{u}\colon
\mathcal{F}(\mathbf{u})\subset \Sigma\})$.  This can be called lower
probability and written $\underbar{\ensuremath{P}}(\Sigma)$.  The upper
probability or ``plausibility'' $\bar{\ensuremath{P}}(\Sigma)$ is defined as
$1-\underbar{\ensuremath{P}}(\Sigma^c)$, or equivalently
$\nu_{\mathbf{x}}(\{\mathbf{u}\colon \mathcal{F}(\mathbf{u}) \cap \Sigma \neq
\emptyset)\}$.  Bayesian inference is recovered exactly when
combining the distribution of $\mathcal{F}(\mathbf{u})$ obtained from $\mathbf{x}$
with a prior distribution on $\theta$,
see \citet{dempster1968generalization} and Section \ref{subsec:partialprior}.
Following \citet{Dempster:IJAR} DS inference can be summarized via the
probability triple $(\pp, \qq, \rr)$:
\begin{equation}
\pp(\Sigma) = \underbar{\ensuremath{P}}(\Sigma), \quad \qq(\Sigma) = 1 - \bar{P}(\Sigma), \quad \rr(\Sigma) = \bar{P}(\Sigma)  - \underbar{\ensuremath{P}}(\Sigma), \label{eq:pqr}
\end{equation}
with $\pp + \qq + \rr = 1$ for all $\Sigma$, quantifying support ``for'',
``against'', and ``don't know'' about the assertion $\Sigma$. As argued in
\citet{Dempster:IJAR,gong2019simultaneous}, the triple $(\pp, \qq,
\rr)$ draws a stochastic parallel to the three-valued logic, with $\rr$
representing weight of evidence in a third, indeterminate logical state. A
$\pp$ or $\qq$ value close to 1 is interpreted as strong evidence towards
$\Sigma$ or $\Sigma^c$, respectively. A large $\rr$ suggests that the model and
data are structurally deficient in making precise judgment about the assertion
$\Sigma$ or its negation.

Sampling methods enable approximations of these probabilities via standard
Monte Carlo arguments.  A simple strategy is to draw $(u_n)_{n\in[N]}$ from the
uniform on $\Delta^N$ until $\mathcal{F}(\mathbf{u})$ is non-empty. 
However the rejection rate would be prohibitively high
as $N$ increases.  
Some properties of $\nu_\mathbf{x}$ have
been obtained in \citet{Dempster66,dempster1972class}.  For example, Equation
(2.1) in \citet{dempster1972class} states that, for a fixed $\theta\in\Delta$,
$\nu_{\mathbf{x}}(\{\mathbf{u}\colon \theta \in \mathcal{F}(\mathbf{u})\})$ is
equal to the Multinomial probability mass function with parameter $\theta$ 
evaluated at $N_1,\ldots,N_K$. Equation (2.5) in \citet{dempster1972class} gives the expected volume
of $\mathcal{F}(\mathbf{u})$.
\citet{dempster1972class} also obtains the distribution of
vertices of $\mathcal{F}(\mathbf{u})$ under $\mathbf{u}\sim\nu_{\mathbf{x}}$ with smallest and largest
coordinate $\theta_k$ for any $k\in[K]$, which are Dirichlet distributions.
These enable the approximation  of $(\pp,\qq,\rr)$ for certain
assertions, including the sets $\{\theta\colon \theta_k\in [0,c]\}$ for
arbitrary $c\in[0,1]$. However, for general assertions 
the joint distribution of all vertices of $\mathcal{F}(\mathbf{u})$ 
under $\mathbf{u}\sim\nu_{\mathbf{x}}$
is necessary, as 
in the case of both applications in Section~\ref{sec:application}.

\section{Proposed Gibbs sampler \label{sec:gibbssampler}}

\begin{algorithm}
    Input: $k\in[K]$, $\theta\in\Delta$, and the vertices of $\Delta$ denoted by $(V_{\ell})_{\ell\in[K]}$.
\begin{itemize}
\item Sample $(w_{1},\ldots,w_{K})$ uniformly on $\Delta$,\\
e.g. $\tilde{w}_{\ell}\sim\text{Exponential}(1)$ for all $\ell\in[K]$
and $w_{\ell}=\tilde{w}_{\ell}/\sum_{j=1}^{K}\tilde{w}_{j}$.
\item Define the point $z=w_{k}\theta+\sum_{\ell\neq k}w_{\ell}V_{\ell}$,\\
e.g. $z_{k}=w_{k}\theta_{k}$ and $z_{\ell}=w_{k}\theta_{\ell}+w_{\ell}$
for $\ell\neq k$.
\item Return $z$, a uniformly distributed point in $\Delta_{k}(\theta)$.
\end{itemize}
\caption{\label{alg:Uniform-sampling-in}Uniform sampling in $\Delta_{k}(\theta)$.}
\end{algorithm}

\subsection{Strategy \label{subsec:generalstrategy}}

The proposed algorithm is a Markov chain Monte Carlo (MCMC) method
targeting $\nu_\mathbf{x}$, thus referred to as the target distribution.
At the initial step, 
we set $\theta^{(0)}$ arbitrarily in $\Delta$, for example a draw from a Dirichlet distribution. 
Given $\theta^{(0)}$ we can sample,
for $k\in[K]$ and $n\in\mathcal{I}_{k}$, $u_{n}\sim \Delta_{k}(\theta^{(0)})$.
Then $\mathbf{u} = (u_{n})_{n\in[N]}$ is in $\mathcal{R}_{\mathbf{x}}$ because
$\theta^{(0)}$ is in $\mathcal{F}(\mathbf{u})$ by construction.
Sampling uniformly over $\Delta_k(\theta)$ can be done following 
equation (5.7) in \citet{Dempster66}, as recalled in Algorithm \ref{alg:Uniform-sampling-in}.
In this section we assume that $N_k = |\mathcal{I}_k|\geq 1$ for all $k\in[K]$,
and describe how to handle empty categories in Section \ref{subsec:addingremoving}.

We draw components of $\mathbf{u}$
from conditional distributions given the other components under $\nu_{\mathbf{x}}$,
namely we draw $\mathbf{u}_{\mathcal{I}_k} = (u_n)_{n\in \mathcal{I}_k}$
for $k\in[K]$ from 
$\nu_{\mathbf{x}}(d\mathbf{u}_{\mathcal{I}_{k}}|\mathbf{u}_{[N]\setminus\mathcal{I}_{k}})$.
Drawing $\mathbf{u}_{\mathcal{I}_k}$ from this conditional
distribution will constitute an iteration of a Gibbs sampler,
illustrated in Figure \ref{fig:gibbs:strategy} for the data $N_1=2,N_2=3,N_3=1$.
Figure \ref{fig:gibbs:1} shows a sample $\mathbf{u}\in\mathcal{R}_\mathbf{x}$,
with each $u_n$ colored according to $x_n\in[K]$.
Sampling from 
$\nu_{\mathbf{x}}(d\mathbf{u}_{\mathcal{I}_{k}}|\mathbf{u}_{[N]\setminus\mathcal{I}_{k}})$
can be understood as drawing all the points
of the same color conditional on the other points.
The overall
Gibbs sampler cycles through the $K$ categories
to generate a sequence of draws $\mathbf{u}^{(t)}$
that converges to $\nu_{\mathbf{x}}$ as $t\to
\infty$, for example in distribution. 
To each $\mathbf{u}^{(t)}$ is associated a feasible set $\mathcal{F}(\mathbf{u}^{(t)})$
that can contribute to the approximation of $(\pp,\qq,\rr)$ triples described in Section \ref{subsec:dsinference}.
The next question is how to sample from
the adequate conditional distributions.
Towards this aim we will draw on a representation of $\mathcal{R}_\mathbf{x}$ 
connected to the presence of ``negative cycles'' in a complete graph with $K$ vertices.

\begin{figure}[t]
    \centering
    \subfloat[\label{fig:gibbs:1}]{\begin{centering}
        \includegraphics[width=0.35\textwidth]{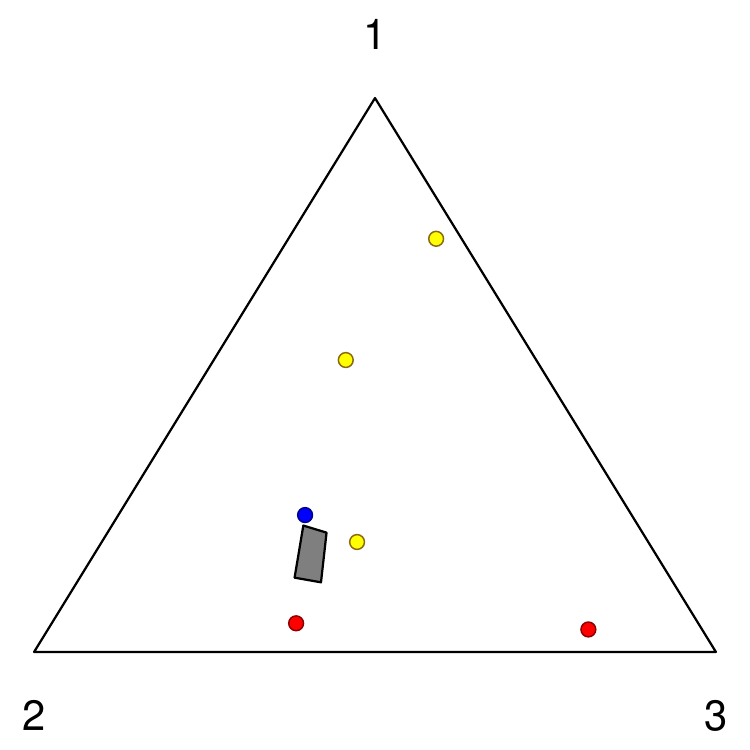}
\par\end{centering} }
\hspace*{1cm}
    \subfloat[\label{fig:gibbs:4}]{\begin{centering}
        \includegraphics[width=0.35\textwidth]{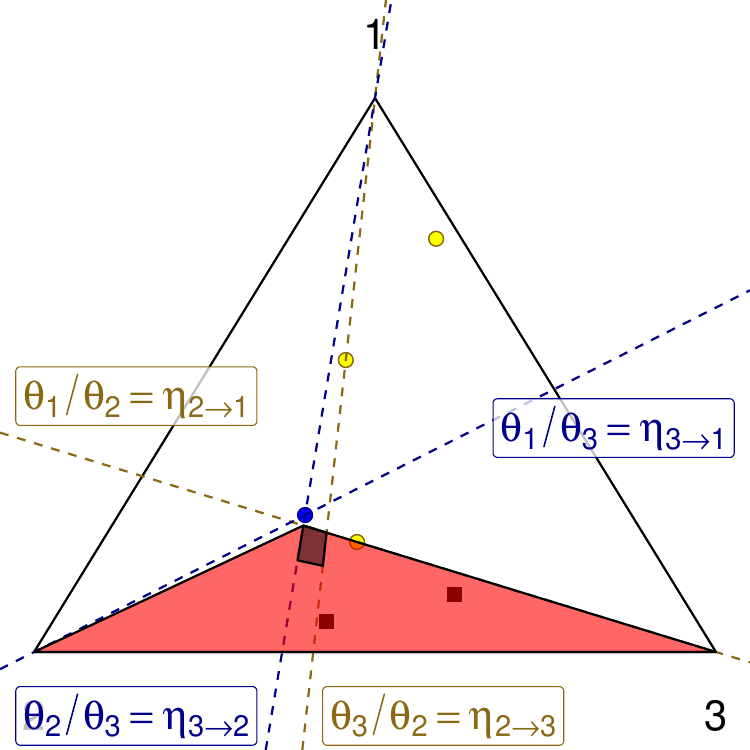}
\par\end{centering} }
\caption{ \label{fig:gibbs:strategy} Given $\mathbf{u}\in\mathcal{R}_{\mathbf{x}}$
(shown in \ref{fig:gibbs:1}), drop components $\mathbf{u}_{\mathcal{I}_k}$ for some $k\in[K]$ 
(the red dots in \ref{fig:gibbs:1}) and draw new components 
$\mathbf{u}_{\mathcal{I}_k}$ (the red squares in \ref{fig:gibbs:4}) from their
conditional distribution identified in Proposition \ref{prop:conditional}
with support represented by the shaded triangle. 
}
\end{figure}

\subsection{Non-emptiness of feasible sets \label{subsec:constraints}}

We can represent $\mathcal{R}_{\mathbf{x}}$ without mention of the
existence of some $\theta \in \Delta$ as in \eqref{eq:RN}, but 
instead with explicit constraints on the components of $\mathbf{u}$.
We first find an equivalent representation of $\theta\in\mathcal{F}(\mathbf{u})$ for a 
fixed $\mathbf{u}$. By definition $\theta\in\mathcal{F}(\mathbf{u})$ satisfies for all
$n\in[N]\;u_{n}\in\Delta_{x_{n}}(\theta)$. For each $k\in[K]$, using Lemma \ref{Lemma-5.2} we write
\[
\forall n\in\mathcal{I}_{k}\quad\forall\ell\in[K]\setminus\{k\}\quad\frac{u_{n,\ell}}{u_{n,k}}\geq\frac{\theta_{\ell}}{\theta_{k}} \quad\Leftrightarrow \quad
\forall\ell\in[K]\setminus\{k\}\quad\min_{n\in\mathcal{I}_{k}}\frac{u_{n,\ell}}{u_{n,k}}\geq\frac{\theta_{\ell}}{\theta_{k}}.
\]
This prompts the definition
\begin{equation}
\forall k\in[K]\quad\forall\ell\in[K] \quad\eta_{k\to\ell}(\mathbf{u})=\min_{n\in\mathcal{I}_{k}}\frac{u_{n,\ell}}{u_{n,k}}.\label{eq:definitionetas}
\end{equation}
Observe that the values $(\eta_{k\to\ell}(\mathbf{u}))$ depend on the observations through the sets $(\mathcal{I}_k)$.
At this point, $\theta\in\mathcal{F}(\mathbf{u})$
is equivalent to $\theta_{\ell}/\theta_{k}\leq\eta_{k\to\ell}(\mathbf{u})$ for
$\ell,k\in[K]$.
Next assume $\theta\in\mathcal{F}(\mathbf{u})$ and consider some 
implications.
First, for all $k,\ell$
\[
\frac{\theta_{\ell}}{\theta_{k}}\leq\eta_{k\to\ell}(\mathbf{u}),\quad\text{and}\quad\frac{\theta_{k}}{\theta_{\ell}}\leq\eta_{\ell\to k}(\mathbf{u}),\quad\text{thus}\quad\eta_{k\to\ell}(\mathbf{u})\eta_{\ell\to k}(\mathbf{u})\geq1.
\]
If $K\geq 3$ we can write $\theta_{\ell}/\theta_{k}$ as $(\theta_{\ell}/\theta_{j})(\theta_{j}/\theta_{k})$,
and apply a similar reasoning to obtain the inequalities
$\eta_{\ell\to k}(\mathbf{u})\eta_{k\to j}(\mathbf{u})\eta_{j\to\ell}(\mathbf{u})\geq 1$  for all $k,\ell,j$.
Overall we can write, for all $K\geq2$, with any number
$L$ of indices $j_{1},\ldots,j_{L}\in[K]$, the following
constraints:
\begin{equation}
	\forall L\in [K] \quad\forall j_{1},\ldots,j_{L}\in[K]\quad\eta_{j_{1}\to j_{2}}(\mathbf{u})
    \eta_{j_{2}\to j_{3}}(\mathbf{u})\ldots\eta_{j_{L}\to j_{1}}(\mathbf{u})\geq1.\label{eq:constraints}
\end{equation}
Hereafter we drop ``$(\mathbf{u})$'' from the notation for clarity.
The case $L=1$  gives inequalities $\eta_{k\to k}\geq1$ which are
always satisfied since $\eta_{k\to k}=1$ following \eqref{eq:definitionetas}.
Furthermore, it suffices to consider only indices $j_{1},\ldots,j_{L}$ that are
unique, otherwise 
the associated inequality in \eqref{eq:constraints} is
implied by inequalities associated with smaller values of $L$.

At this point, we observe a fruitful connection between
\eqref{eq:constraints} and directed graphs.
The indices in $[K]$ can be viewed as vertices of a fully connected directed
graph. Directed edges are ordered pairs $(j_{1},j_{2})$. We
associate the product $\eta_{j_{1}\to j_{2}}\eta_{j_{2}\to
j_{3}}\ldots\eta_{j_{L}\to j_{1}}$ with a sequence of edges, $(j_{1},j_{2})$,
$(j_{2},j_{3})$, up to $(j_{L},j_{1})$. That sequence forms a path
from vertex $j_{1}$ back to vertex $j_{1}$, of length $L$, in other words
a directed cycle of length $L$.
Define $w_{k\to\ell}=\log\eta_{k\to\ell}$ for all $k,\ell\in[K]$, and 
treat it as the weight of edge $(k,\ell)$. Then the inequality \eqref{eq:constraints} 
is equivalent to
$w_{j_{1}\to j_{2}}+w_{j_{2}\to j_{3}}+\ldots+w_{j_{L}\to j_{1}}\geq0$. The sum
of weights along a path is called its ``value''. 
The inequalities in \eqref{eq:constraints} are then equivalent to 
all cycles in the graph having non-negative values. See Figure \ref{fig:constraints} for an illustration
for $K = 3$ of the equivalent conception of constraints in
\eqref{eq:constraints}  as graph cycle values. Detecting whether graphs
contain cycles with negative values, called ``negative cycles'', can be done with the
Bellman--Ford algorithm \citep{bang2008digraphs}. 

At this point we have established that 
$\theta\in\mathcal{F}(\mathbf{u})$ implies the inequalities of
\eqref{eq:constraints}, which can be understood as 
constraints on the weights of a graph. Our next result states that the converse also holds. 
\begin{prop}
\label{prop:existence}There exists $\theta\in\Delta$ satisfying
$\theta_{\ell}/\theta_{k}\leq \eta_{k\to\ell}$ for all $k,\ell\in[K]$
if and only if the values $(\eta_{k\to\ell})$ satisfy 
\begin{equation}
	\forall L\in [K] \quad\forall j_{1},\ldots,j_{L}\in[K]\quad\eta_{j_{1}\to j_{2}}\eta_{j_{2}\to j_{3}}\ldots\eta_{j_{L}\to j_{1}}\geq1.\label{eq:inequalities:inproposition}
\end{equation}
Furthermore it suffices to restrict \eqref{eq:inequalities:inproposition} to
distinct indices $j_1,\ldots,j_L$.
\end{prop}

\begin{proof}
The proof of the reverse
implication explicitly constructs a feasible $\theta$ based on the
values $(\eta_{k\to\ell})$, assuming that they satisfy \eqref{eq:inequalities:inproposition}.
Introduce the fully connected graph with $K$ vertices, with 
weight 
$\log\eta_{k\to\ell}$ on edge $(k,\ell)$. 
Thanks to $(\eta_{k\to\ell})$ satisfying \eqref{eq:inequalities:inproposition},
there are no negative cycles thus one cannot
decrease the value of a path by appending a cycle to it.
Since there are only finitely many paths without cycles 
there is a finite minimal value over all paths from $k$ to
$\ell$, which we denote by $\min(k\to\ell)$. 
In other words \eqref{eq:inequalities:inproposition}
implies that $\min(k\to\ell)$ is finite.

We choose a vertex in $[K]$ arbitrarily, for
instance vertex $K$. We define $\theta$ by
$\theta_{k}=\exp(\min(K\to k))$ and 
then by normalizing the entries so that $\theta\in\Delta$.
We can write
$\min(K\to\ell)\leq\min(K\to k)+\log\eta_{k\to\ell}$,
because the right hand side is the value of a path from $K$ to $\ell$
(via $k$), while the left hand side is the smallest value over all
such paths. Upon taking the exponential, the above
inequality is equivalent to $\theta_{\ell}/\theta_{k}\leq\eta_{k\to\ell}$.
\end{proof}

\begin{figure}[t]
    \centering
    \subfloat[\label{fig:constraints1}]{\begin{centering}
        \includegraphics[width=0.35\textwidth]{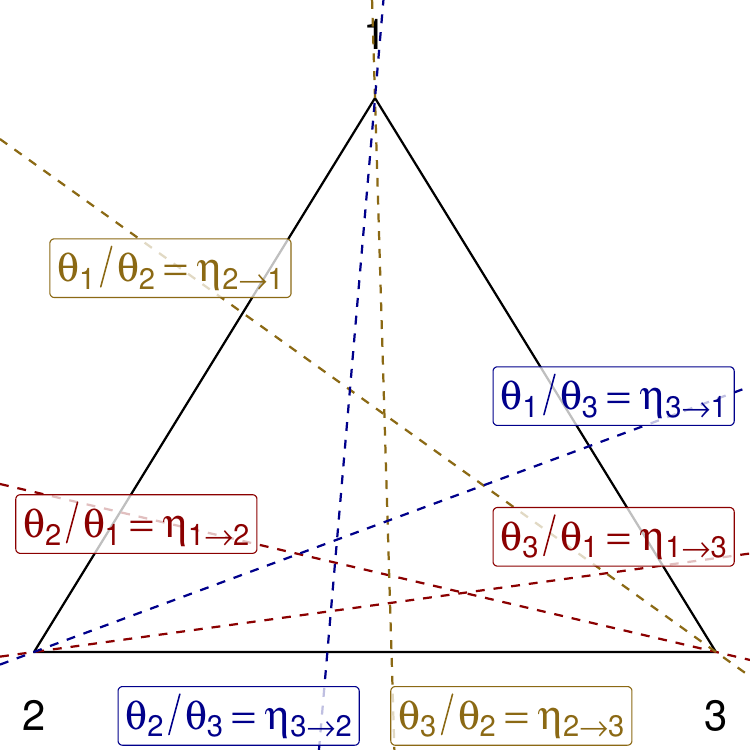}
\par\end{centering} }
\hspace*{1cm}
    \subfloat[\label{fig:constraints2}]{\begin{centering}
        \includegraphics[width=0.35\textwidth]{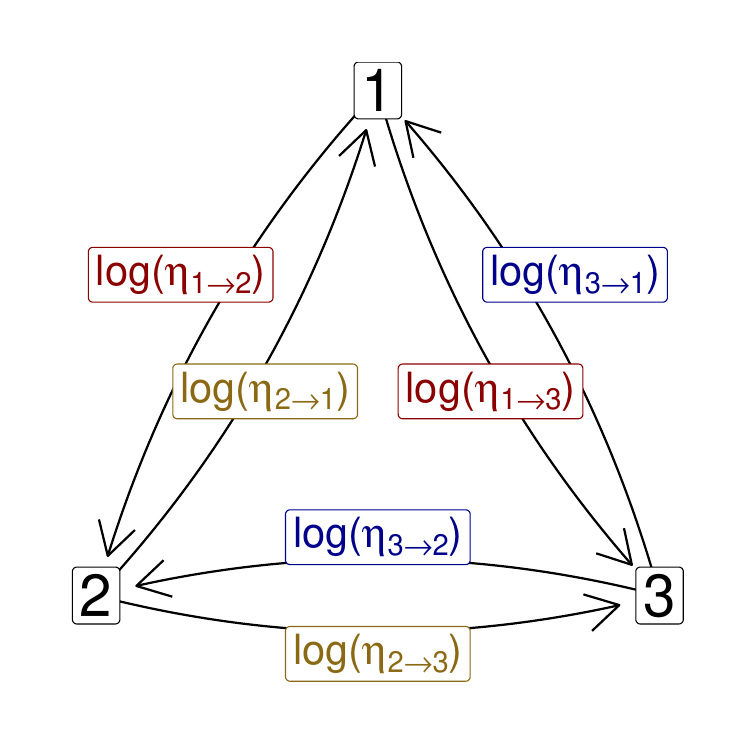}
\par\end{centering} }
\caption{ \label{fig:constraints} Two views on the constraints in
\eqref{eq:constraints}. In \ref{fig:constraints1} the values $\eta_{k\to\ell}$
define linear constraints $\theta_\ell/\theta_k = \eta_{k\to\ell}$.
In \ref{fig:constraints2} the log values are weights on the edges of a complete directed graph.}
\end{figure}

\subsection{Conditional distributions \label{subsec:conditionaldistribution}}

Thanks to Proposition \ref{prop:existence} we can write
$\mathcal{R}_{\mathbf{x}}$ defined in \eqref{eq:RN} as the set of $\mathbf{u}$
for which the values $\eta_{k\to\ell}$ satisfy
\eqref{eq:inequalities:inproposition}, with $\eta_{k\to\ell}$ defined in
\eqref{eq:definitionetas}.  We next provide a representation of the conditional
distribution of $\mathbf{u}_{\mathcal{I}_k}$ under $\nu_{\mathbf{x}}$ that is
convenient for sampling purposes.

\begin{prop}
\label{prop:conditional}
Let $\mathbf{u}=(u_{1},\ldots,u_{N})\in\mathcal{R}_{\mathbf{x}}$,
and define $\eta_{k\to\ell}=\min_{n\in\mathcal{I}_{k}}u_{n,\ell}/u_{n,k}$
for all $k,\ell\in[K]$. Let $k\in[K]$.
Define for $\ell \in [K]$,
\begin{equation}
    \label{eq:conditionaltheta}
\theta_\ell = \frac{\exp(-\min(\ell \to k))}{\sum_{\ell'\in[K]} \exp(-\min(\ell' \to k))}
\end{equation}
where $\min(\ell\to k)$ is the minimum value over all paths from $\ell$ to $k$,
in a fully connected directed graph with weight
$\log\eta_{j\to\ell}$ on edge $(j,\ell)$.
Then, $\nu_\mathbf{x}(d\mathbf{u}_{\mathcal{I}_k}|\mathbf{u}_{[N]\setminus\mathcal{I}_k})$ is 
the uniform distribution on $\Delta_{k}(\theta)^{N_k}$.
\end{prop}

In other words
$\nu_\mathbf{x}(d\mathbf{u}_{\mathcal{I}_k}|\mathbf{u}_{[N]\setminus\mathcal{I}_k})$
is the product measure with each component $u_n$ following the uniform
distribution on $\Delta_k(\theta)$, with $\theta$
defined in \eqref{eq:conditionaltheta}. The proposition 
is key to the implementation of the proposed Gibbs sampler.

\begin{proof}
We consider an arbitrary $k\in[K]$,
and assume that $\mathbf{u}\in\mathcal{R}_{\mathbf{x}}$. 
Listing the inequalities in \eqref{eq:inequalities:inproposition} that involve the index $k$ 
and separating the terms $\eta_{k\to\ell}$ from the others, 
we obtain 
\begin{align}
    \forall \ell \in [K] \quad &\quad \eta_{k\to\ell} \geq \eta_{\ell\to k}^{-1} \; ,  \label{eq:refreshconstraint1}\\
    \forall \ell \in [K] \quad \forall L\in [K-2]\quad\forall j_{1},\ldots,j_{L}\in[K]\setminus\{k,\ell\} & \quad\eta_{k\to\ell}\geq\left(\eta_{\ell\to j_{1}}\ldots\eta_{j_{L}\to k}\right)^{-1}.\label{eq:refreshconstraint2}
\end{align}
Thus, for $\mathbf{u}$ to remain in $\mathcal{R}_{\mathbf{x}}$ after updating
its components $\mathbf{u}_{\mathcal{I}_{k}}$,
it is enough to check that the ratios $u_{n,\ell}/u_{n,k}$
for $\ell\in[K]$ and $n\in \mathcal{I}_k$ are lower bounded as above.

The finiteness of $\min(\ell \to k)$
results from the same reasoning as in the proof of Proposition \ref{prop:existence}.
Note that $\min(\ell \to k)$ can be constructed without
the entries $\mathbf{u}_{\mathcal{I}_k}$ of $\mathbf{u}$, because the shortest
path from $\ell$ to $k$ should pay no attention to any directed edges that stem
from $k$, and the entries $\mathbf{u}_{\mathcal{I}_k}$ inform only the weights
of edges stemming from $k$.
Thus we can define $\theta$ as in \eqref{eq:conditionaltheta}.

We next show that the support of $\nu_\mathbf{x}(d\mathbf{u}_{\mathcal{I}_k}|\mathbf{u}_{[N]\setminus\mathcal{I}_k})$ 
is exactly $\Delta_k(\theta)^{N_k}$.
Let $\mathbf{u}_{\mathcal{I}_k} \in \Delta_k(\theta)^{N_k}$. By Lemma \ref{Lemma-5.2} 
and the definition of $\theta$, we have
\[\forall \ell \in [K] \quad \min_{n\in\mathcal{I}_k} \frac{u_{n,\ell}}{u_{n,k}} \geq \exp(-\min(\ell \to k)).\]
But $\exp(-\min(\ell \to k))=(\exp(\min(\ell\to k)))^{-1}$ is greater than $(\eta_{\ell\to j_1}\ldots \eta_{j_L\to k})^{-1}$
for any path $\ell\to j_1\ldots j_L\to k$. Thus, with $\eta_{k\to \ell} = \min_{n\in\mathcal{I}_k} u_{n,\ell}/u_{n,k}$,
inequalities \eqref{eq:refreshconstraint1}-\eqref{eq:refreshconstraint2} are
satisfied. Proposition \ref{prop:existence} 
guarantees that $\mathbf{u}$ is in $\mathcal{R}_\mathbf{x}$, thus $\Delta_k(\theta)^{N_k}$ 
is contained in the support
of $\nu_\mathbf{x}(d\mathbf{u}_{\mathcal{I}_k}|\mathbf{u}_{[N]\setminus\mathcal{I}_k})$.

Let us show the reverse inclusion by 
considering $\mathbf{u}_{\mathcal{I}_k} \notin \Delta_k(\theta)^{N_k}$.
There, for some $n\in \mathcal{I}_k$ and some $\ell\in[K]$, we have $u_{n,\ell}/u_{n,k}< \exp(-\min(\ell \to k))$.
Denote by $\ell \to j_1\ldots j_L\to k$ the path attaining the value $\min(\ell \to k)$. We obtain 
$\eta_{k\to \ell} \leq u_{n,\ell}/u_{n,k} < (\eta_{\ell\to j_{1}}\ldots\eta_{j_{L}\to k})^{-1}$,
and thus $\eta_{k\to \ell}\eta_{\ell\to j_{1}}\ldots\eta_{j_{L}\to k}<1$, in other words
some inequalities in \eqref{eq:inequalities:inproposition} are not satisfied and thus,
by Proposition \ref{prop:existence}, $\mathbf{u}$ is not in $\mathcal{R}_\mathbf{x}$.
\end{proof}

Proposition \ref{prop:conditional} 
provides a strategy to sample from 
the conditional distributions of interest,
provided that we can obtain $\theta\in\Delta$
in \eqref{eq:conditionaltheta},
which involves $\min(\ell\to k)$ for all $\ell$.
These can be obtained from shortest path algorithms such
as Bellman--Ford implemented
in \texttt{igraph} \citep{igraph}.
Alternatively we can view
$\theta$ in \eqref{eq:conditionaltheta}
as the solution of the linear program,
\begin{align}
	\text{max} \left\{\theta_{k}\colon
    \theta\in\Delta \quad\forall j\neq k\quad\forall i\neq j\quad
    \frac{\theta_{i}}{\theta_{j}}\leq\eta_{j\to i}\right\}. \label{eq:LP}
\end{align}
This has a simple interpretation: $\theta$ in \eqref{eq:conditionaltheta}
is precisely the vertex of $\mathcal{F}(\mathbf{u})$ with the largest $k$-th component.
The equivalence between shortest path problems and linear programs 
is well known. Implementations are provided 
in \texttt{lpsolve} \citep{berkelaar2004lpsolve,lpsolve2014lpsolveapi}. 

The Gibbs sampler is described in Algorithm \ref{alg:Gibbs}. 
Its outputs include $\mathbf{u}^{(t)}$
converging to $\nu_{\mathbf{x}}$ in distribution as $t\to\infty$,
as well as the associated values 
of $(\eta_{k\to\ell}^{(t)})$ from which we can obtain
the sets $\mathcal{F}(\mathbf{u}^{(t)})$
as
$\{\theta\in\Delta: \theta_\ell/\theta_k \leq \eta_{k\to\ell}^{(t)} \;\forall k,\ell\in[K]\}$.
Such sets can be stored in ``half-space representation''
or as a list of vertices in $\Delta$, obtained 
by vertex enumeration \citep{avis1992pivoting}. 
Convenient functions to store and manipulate polytopes can be found
in \texttt{rcdd} \citep{geyer2008r,fukuda1997cdd}.
We run  $100$ iterations of
the sampler and record elapsed seconds for different values of $N$ and $K$. Medians
over $50$ experiments are reported in Figure \ref{fig:Computational-time},
for counts set to $\lfloor N/K\rfloor$ in each category. 

\begin{figure}[t]
    \centering
    \subfloat[\label{fig:scalingwithk}]{\begin{centering}
        \includegraphics[width=0.35\textwidth]{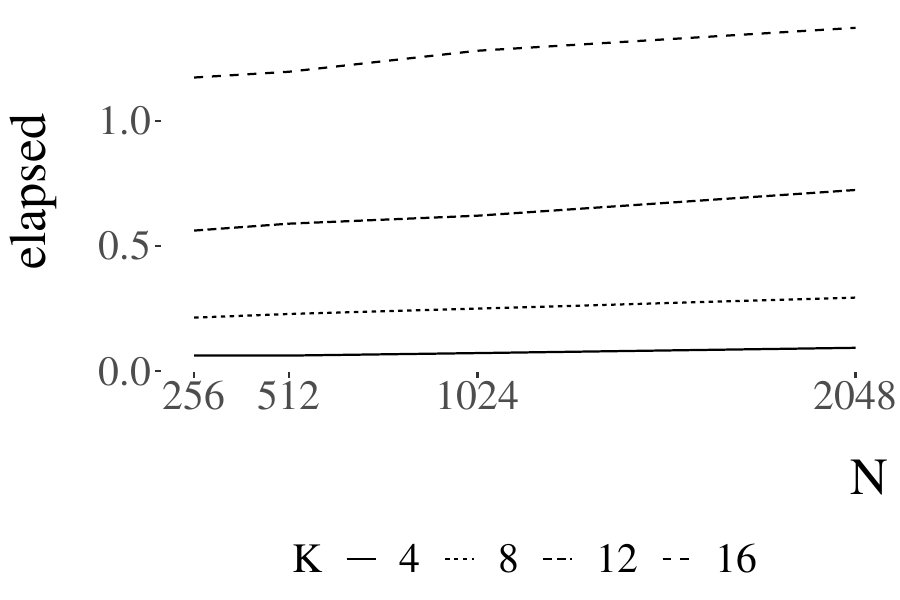}
\par\end{centering} }
\hspace*{1cm}
    \subfloat[\label{fig:scalingwithN}]{\begin{centering}
        \includegraphics[width=0.35\textwidth]{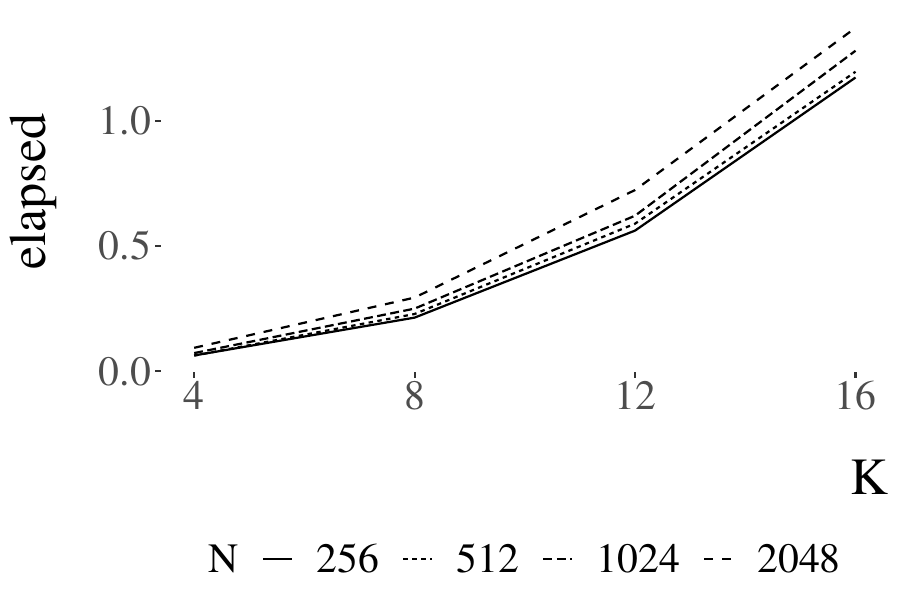}
\par\end{centering} }
\caption{\label{fig:Computational-time} Elapsed time in seconds for
100 iterations of the sampler.
In \ref{fig:scalingwithk}, elapsed time as a function of $N$, for different $K$. 
In \ref{fig:scalingwithN}, elapsed time
as a function of $K$, for different $N$.}
\end{figure}

\begin{algorithm}
	\begin{enumerate}
		\item Set $\theta^{(0)}$ in $\Delta$, 
        and for all $k\in[K]$, all $n\in\mathcal{I}_k$, sample $u_n^{(0)} \sim \Delta_k(\theta^{(0)})$
			(Algorithm \ref{alg:Uniform-sampling-in}). 
		\item Compute $\eta_{k\to\ell}^{(0)} = \min_{n\in\mathcal{I}_k} u_{n,\ell}^{(0)} / u_{n,k}^{(0)}$ for all $k,\ell\in[K]$.
		\item At iteration $t \geq 1$, 
			\begin{enumerate}
				\item Set $\eta_{k\to\ell}^{(t)} \leftarrow \eta_{k\to\ell}^{(t-1)}$ for all $k,\ell\in[K]$.
				\item For category $k\in[K]$,
					\begin{enumerate}
						\item Compute $\theta\in\Delta$ from the values $(\eta_{j\to\ell}^{(t)})$
                            according to \eqref{eq:conditionaltheta},
                        \item[] either by computing shortest paths 
                        or by solving a linear program \eqref{eq:LP}.
                        \item For each $n\in\mathcal{I}_{k}$, sample
                            $u^{(t)}_{n} \sim \Delta_{k}(\theta)$
							(Algorithm \ref{alg:Uniform-sampling-in}).
						\item Set $\eta_{k\to\ell}^{(t)}\leftarrow\min_{n\in\mathcal{I}_{k}}u^{(t)}_{n,\ell}/u^{(t)}_{n,k}$,
							for all $\ell\neq k$.
					\end{enumerate}
			\end{enumerate}
	\end{enumerate}
	\caption{\label{alg:Gibbs} Gibbs sampler for Categorical inference in the Dempster--Shafer framework. \protect \\
    Input: observations $\mathbf{x}\in[K]^N$, defining index sets $\mathcal{I}_k=\{n\in[N]: x_n = k\}$ for $k\in[K]$.\\
    Output: sequence $(\mathbf{u}^{(t)})_{t\geq 0}$ converging to $\nu_\mathbf{x}$, the uniform
    distribution on $\mathcal{R}_\mathbf{x}$.}
\end{algorithm}

\subsection{Convergence to stationarity\label{subsec:convergence}}

A common question to all MCMC algorithms is the rate of convergence to
stationary \citep{jerrum1998mathematical,roberts2004general}, which here might depend
on $K$ and the observed counts $N_1,\ldots,N_K$.  In the simplest case where
$K=2$, with counts $N_1\geq 1, N_2\geq 1$, we obtain an upper bound on the
mixing time of the chain in the 1-Wasserstein metric \citep[e.g.][]{gibbs2004}
as detailed in Appendix \ref{appx:cvgrate}.  We find that the upper bound
increases at most linearly with the total count $N=N_1+N_2$.  The extension of
this theoretical result to arbitrary $K\geq 3$ is left as an
open question.

For arbitrary $K$ we use the empirical approach of
\citet{biswas2019estimating}, that provides estimated upper bounds on the total
variation distance (TV) between $\mathbf{u}^{(t)}$ at iteration $t$ and
$\nu_{\mathbf{x}}$. These upper bounds are obtained as empirical averages over
independent runs of coupled Markov chains; see Appendix
\ref{appx:empiricalcvgrate} for a brief description of the approach.
For $K=4,8,12,16$, we construct synthetic data sets with $10$ observations in each category
and estimate upper bounds for a range of $t$ shown in Figure
\ref{fig:mixingwithk}.  The number of iterations required
for convergence seems to be stable in $K$.
Next, we set $K=5$ and consider $10,20,30,40$ counts in each category,
leading to $N$ varying between $50$ and $200$. Figure \ref{fig:mixingwithN}
shows the associated upper bounds, that increase with $N$.

\begin{figure}[t]
    \centering
    \subfloat[\label{fig:mixingwithk}]{\begin{centering}
        \includegraphics[width=0.35\textwidth]{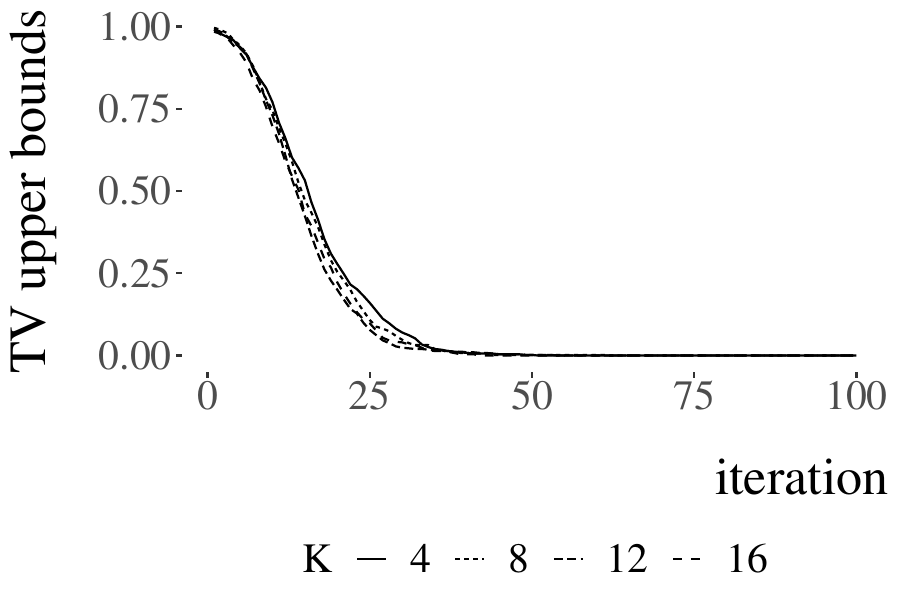}
\par\end{centering} }
\hspace*{1cm}
    \subfloat[\label{fig:mixingwithN}]{\begin{centering}
        \includegraphics[width=0.35\textwidth]{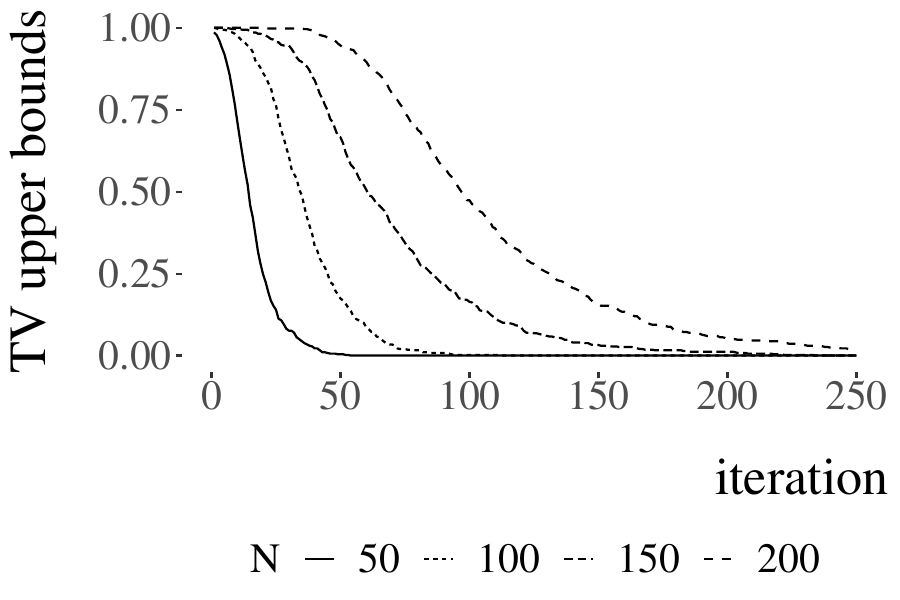}
\par\end{centering} }
\caption{\label{fig:mixing-time}Upper bounds on the TV distance
between $\mathbf{u}^{(t)}$  and
$\nu_{\mathbf{x}}$ against $t$. 
Figure \ref{fig:mixingwithk} shows varying $K$ with
$10$ counts in each category. 
Figure \ref{fig:mixingwithN} shows varying 
$N$ with $K=5$ and $N/K$ counts in each category.}
\end{figure}

\section{Adding categories, observations and priors \label{sec:manipulations}}

\subsection{Adding empty categories\label{subsec:addingremoving}}

We describe how to add empty and remove empty categories
based on the output of the Gibbs sampler.
Suppose that we have draws $\mathbf{u}$ distributed according to
the target $\nu_{\mathbf{x}}$
associated with a data set $\mathbf{x}\in[K]^{N}$ with $K$ non-empty categories.  
We add a category $K+1$ with $\mathcal{I}_{K+1}=\emptyset$, $N_{K+1}=0$,
and consider how to obtain samples $\mathbf{u}'$ from the corresponding target $\nu_{\mathbf{x}'}$.
Recall that a variable $(u_{1},\ldots,u_{K})$ following
Dirichlet$(1,\ldots,1)$
is equal in distribution to the vector with $\ell$-th entry $w_{\ell}/\sum_{j\in[K]}w_{j}$
for $\ell\in[K]$, where $(w_{\ell})_{\ell\in[K]}$ are independent
Exponential(1). Given $(u_{1},\ldots,u_{K})\sim \Delta$
consider the following procedure.
First, draw $s\sim\text{Gamma}(K,1)$, define $w_{\ell}=s\times u_{\ell}$ for $\ell\in[K]$, and draw $w_{K+1}\sim\text{Exponential}(1)$.
Then define $u'_{\ell}=w_{\ell}/\sum_{j\in[K+1]}w_{j}$ for $\ell\in[K+1]$.
The resulting vector $\mathbf{u}'=(u'_{1},\ldots,u'_{K+1})$ is uniformly
distributed on the probability simplex with $K+1$ vertices denoted by $\Delta'$. 
Since $u'_{\ell}/u'_{k}=u_{\ell}/u_{k}$
for all $k,\ell\in[K]$,
if $(u_{1},\ldots,u_{K})$
satisfies certain constraints on ratios $u_{\ell}/u_{k}$, the same constraints
are satisfied for $(u'_{1},\ldots,u'_{K+1})$. Thus $\mathbf{u}'\sim \nu_{\mathbf{x}'}$.

We can also remove empty categories. Assume that category $K+1$
is empty and that we have draws $\mathbf{u}'\sim\nu_{\mathbf{x}'}$.
For each $u'_{n}$, drop the $(K+1)$-th component $u'_{n,K+1}$, and
define $u_n$ by normalizing the remaining $K$ components. 
The resulting $\mathbf{u}$ follows $\nu_{\mathbf{x}}$.
Importantly, inferences obtained from $\nu_{\mathbf{x}}$ are not 
necessarily identical to
those obtained from $\nu_{\mathbf{x}'}$.  This is illustrated
with Figure \ref{fig:emptycategory}, showing the $(\pp,\qq,\rr)$
probabilities associated with the sets $\{\theta\colon \theta_1\in [0,c)\}$
and $\{\theta\colon \log \theta_1/\theta_2 \in (-\infty,c)\}$,
for the counts $(4,3)$ and $(4,3,0)$.

\begin{figure}[t]
    \centering
    \subfloat[\label{fig:ec1}]{\begin{centering}
        \includegraphics[width=0.35\textwidth]{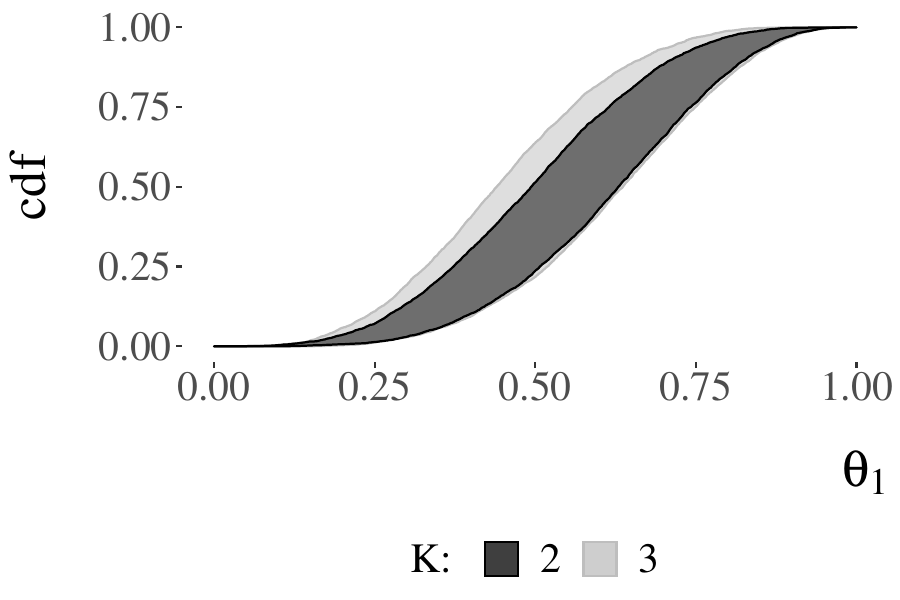}
\par\end{centering} }
\hspace*{1cm}
    \subfloat[\label{fig:ec2}]{\begin{centering}
        \includegraphics[width=0.35\textwidth]{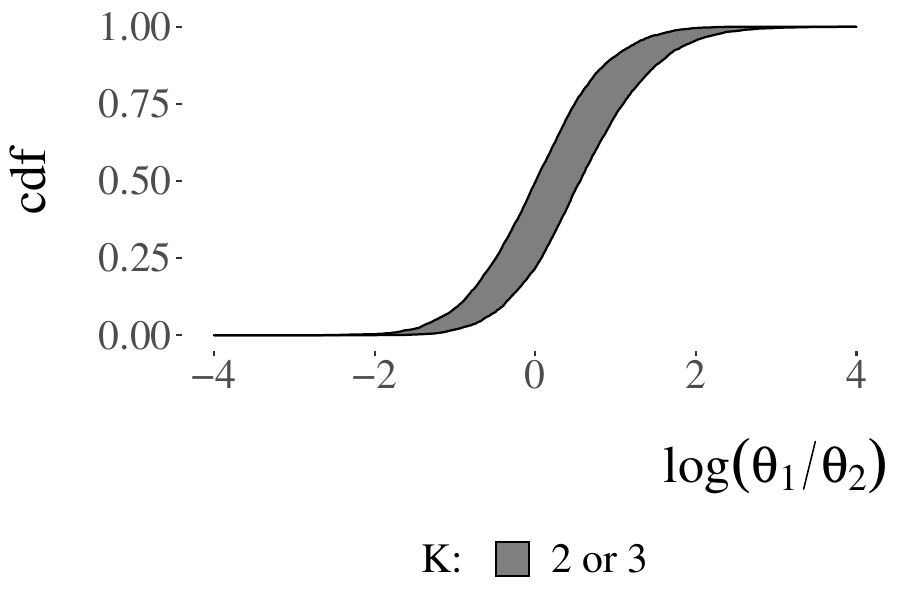}
\par\end{centering} }
\caption{ \label{fig:emptycategory} Inference on $\theta_1$ (\ref{fig:ec1}) and
on $\log (\theta_1/\theta_2)$ (\ref{fig:ec2}) using counts $(4,3)$ ($K=2$) and $(4,3,0)$ ($K=3$). Including an empty third category modifies the inference on $\theta_1$ but not on $\theta_1/\theta_2$.}
\end{figure}

\subsection{Adding partial prior information\label{subsec:partialprior}}

In the DS framework multiple sources of information can be merged
using Dempster's rule of
combination (\citet[Section 5,][]{dempster1967upper}; \citet[Section 2,][]{wasserman1990belief}).
If two sources yield random sets $\mathcal{F}$ and $\mathcal{G}$
the combination is obtained by intersections 
$\mathcal{F}\cap \mathcal{G}$,
under an independent coupling
of $\mathcal{F}$ and $\mathcal{G}$ conditional on $\mathcal{F}\cap \mathcal{G}\neq \emptyset$.
The rule of combination can be used to 
incorporate prior knowledge.
If the prior is encoded as a probability distribution on $\theta\in\Delta$,
we can view each prior draw as a singleton $\mathcal{G}$,
thus intersections $\mathcal{F}\cap \mathcal{G}$ are either singletons or empty.
It can be checked that the non-empty $\mathcal{F}\cap \mathcal{G}$
are equivalent to draws from the posterior 
by noting that, for a given $\theta\in\Delta$, 
\[
    \nu_{\mathbf{x}}\left(\{\mathbf{u}\colon
    \theta\in\mathcal{F}(\mathbf{u})\}\right)=\frac{\text{uniform}\left(\{(u_{1},\ldots,u_{N})\in\Delta^{N}:\theta\in\mathcal{F}(\mathbf{u})\}\right)}{\text{uniform}\left(\{(u_{1},\ldots,u_{N})\in\Delta^{N}:\mathcal{F}(\mathbf{u})\neq\emptyset\}\right)}\propto\theta_{1}^{N_{1}}\ldots\theta_{K}^{N_{K}},
\]
which is proportional to the Multinomial likelihood associated with
$\theta$ and $(N_{1},\ldots,N_{K})$ \citep{dempster1972class}. This justifies why 
DS can be seen as a generalization of Bayesian inference.
In $(\pp,\qq,\rr)$ for an assertion $\Sigma\in \mathcal{B}(\Delta)$ 
this leads to $\pp = \mathbb{P}(\theta\in\Sigma|\mathbf{x})$, the posterior mass 
of $\Sigma$, $\qq=1-\pp$ and $\rr=0$.

The DS framework allows the inclusion of partial prior information.  
We follow the above reasoning except that the prior is formulated as
random sets that are not necessarily singletons.  For example, we can specify a
prior on some components of $\theta$ and extend these into random subsets of
$\Delta$ by ``up-projection'' \citep{Dempster:IJAR} or ``minimal extension''
\citep[Section 2.5,][]{wasserman1990belief}. Concretely suppose that we observe
counts $(N_{1},N_{2})$ of two categories. We specify a Dirichlet prior on
$(\theta_{1},\theta_{2})$ and obtain a Dirichlet posterior.  Next we are told
that there exists in fact a third category, which we could not observe before.
This is different than being told that there is a new
category with zero counts, $N_3 = 0$, which we could handle as in Section
\ref{subsec:addingremoving}.  Up-projection of each posterior draw
$(\theta_1,\theta_2)$ onto the 3-simplex $\Delta$ goes as follows. We compute
$\eta_{1\to 2} = \theta_2/\theta_1$ and $\eta_{2\to 1} = \theta_1/\theta_2$,
and set $\eta_{3\to k} = \eta_{k\to 3}=+\infty$ for $k=1,2$.  Denote by
$\mathcal{F}$ the resulting feasible sets $\{\theta\in\Delta\colon
\theta_\ell/\theta_k \leq \eta_{k\to\ell} \; \forall k,\ell \}$.  These
sets $\mathcal{F}$ correspond to a ``minimal extension'' in
that inference on $\theta_1/\theta_2$ is unchanged,
while inference on $\theta_3$ is vacuous. Vacuous means that
for any assertion $\Sigma = \{\theta\in\Delta \colon \theta_3 \in A\}$ 
with $A\subset [0,1]$, the sets $\mathcal{F}$ result in $\pp=0,\qq=0,\rr=1$.
Using the rule of combination we can subsequently intersect such sets
$\mathcal{F}$ with independent random sets corresponding to new observations of
the three categories.  Visuals are provided in Figure \ref{fig:partialprior},
with \ref{fig:pp2} showing random sets corresponding to counts of
three categories using a partial Dirichlet(2,2) prior on
$(\theta_1,\theta_2)$.

\begin{figure}[t]
    \centering
    \subfloat[\label{fig:pp1}]{\begin{centering}
        \includegraphics[width=0.35\textwidth]{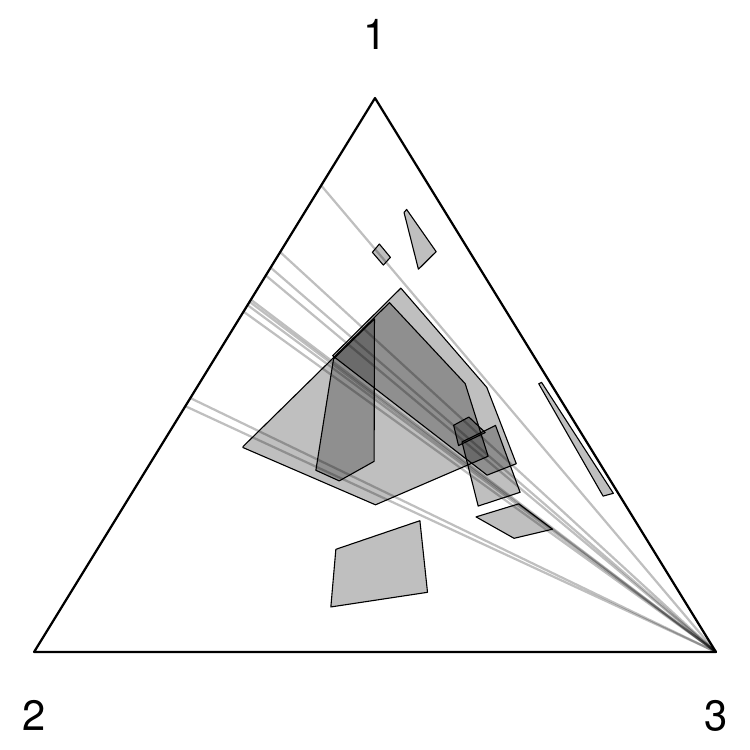}
\par\end{centering} }
\hspace*{1cm}
    \subfloat[\label{fig:pp2}]{\begin{centering}
        \includegraphics[width=0.35\textwidth]{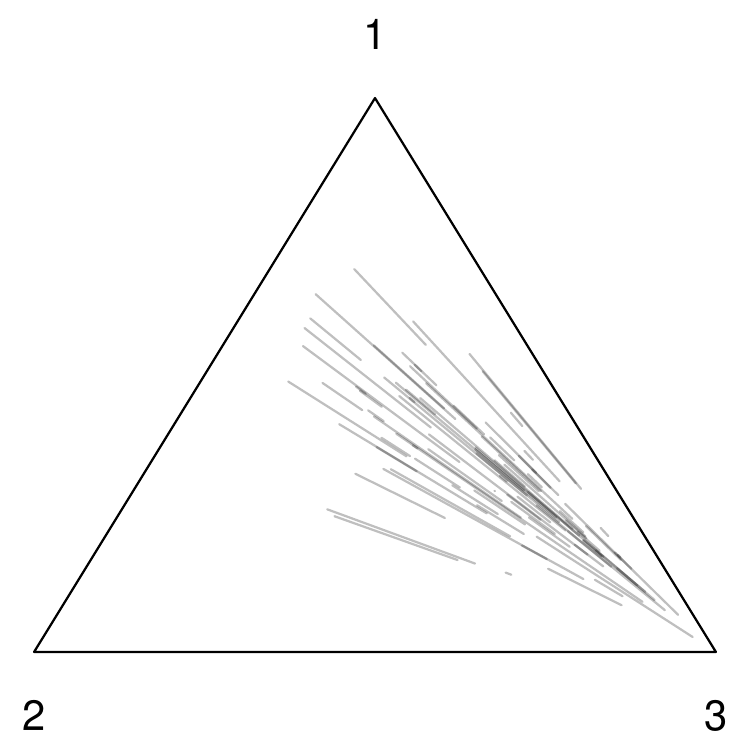}
\par\end{centering} }
\caption{ \label{fig:partialprior} Up-projection  
of posterior samples $(\theta_1,\theta_2)$,
obtained from $(N_1=8,N_2=4)$ and a Dirichlet(2,2) prior
(segments in \ref{fig:pp1}),
and feasible sets obtained independently for counts $(2,1,3)$ (polygons in \ref{fig:pp1}).
The rule of combination retains non-empty intersections of these sets (\ref{fig:pp2}). }
\end{figure}

\subsection{Adding observations \label{subsec:addingobservations}}

We consider the addition of new observations to existing categories.
Denote by $\mathbf{x}_{N+1}$ the original data $\mathbf{x}_{N}$ augmented with 
an observation $x_{N+1}$, which we assume equal to $k\in[K]$.

Any $u_{1:N+1}\in\mathcal{R}_{\mathbf{x}_{N+1}}$
is such that $u_{1:N}\in\mathcal{R}_{\mathbf{x}_N}$ and $u_{N+1}\in\Delta_{k}(\theta)$,
with $\theta\in\Delta$ constructed from $u_{1:N}$ as in Proposition \ref{prop:conditional}. 
Indeed if $u_{1:N+1}=(u_{1},\ldots,u_{N+1}) \in \mathcal{R}_{\mathbf{x}_{N+1}}$,
there exists $\theta'\in\Delta$ such that, for all $n\in[N+1]$, $u_{n,\ell}/u_{n,k}\geq\theta'_{\ell}/\theta'_{k}$.
Thus $u_{1:N} \in \mathcal{R}_{\mathbf{x}_N}$.
We can check that $u_{N+1}$ is in $\Delta_{k}(\theta)$.
Since $u_{1:N+1}\in \mathcal{R}_{\mathbf{x}_{N+1}}$,
then $(u_{n})_{n\in\mathcal{I}_k}$ belongs to the support of $\nu_{\mathbf{x}_{N+1}}(d\mathbf{u}_{\mathcal{I}_k} | \mathbf{u}_{[N+1]\setminus\mathcal{I}_k})$,
which is $\Delta_k(\theta)^{N_k}$ by Proposition \ref{prop:conditional}.
Here we have re-defined $\mathcal{I}_k = \{n\in[N+1]: x_n = k\}$.
Conversely, if $u_{1:N}\in\mathcal{R}_{\mathbf{x}_N}$ and
$u_{N+1}\in\Delta_{k}(\theta)$  then $u_{1:N+1}\in\mathcal{R}_{\mathbf{x}_{N+1}}$; again 
because 
$\Delta_{k}(\theta)$ is precisely the support of
$\nu_{\mathbf{x}_{N+1}}(du_{N+1} | \mathbf{u}_{[N+1]\setminus\mathcal{I}_k})$.

This motivates an importance sampling strategy. 
For $u_{1:N} \sim \nu_{\mathbf{x}_N}$, generate $u_{N+1} \sim \Delta_k(\theta)$,
with $\theta\in\Delta$ as above.
Denote this distribution by $q_{N+1}(du_{N+1}|u_{1:N})$.
The density $u_{N+1}\mapsto q_{N+1}(u_{N+1}|u_{1:N})$
equals $(\theta_k)^{-1}$ for $u_{N+1}\in \Delta_k(\theta)$,
since the volume of $\Delta_k(\theta)$ is $\theta_k$.
We can correct for the discrepancy between proposal and target by 
computing weights 
\begin{align*}
    w_{N+1}(u_{1:N+1})=\frac{\nu_{\mathbf{x}_{N+1}}(u_{1:N+1})}{\nu_{\mathbf{x}_N}(u_{1:N})q_{N+1}(u_{N+1}|u_{1:N})} &
    =\frac{Z_N}{Z_{N+1}}\text{Vol}(\Delta_{k}(\theta)),
\end{align*}
where $Z_N$ is the volume of $\mathcal{R}_{\mathbf{x}_N}$.
We can thus implement self-normalized importance sampling \citep{owenmcbook}.
The reasoning can be extended to assimilate observations 
recursively with a sequential Monte Carlo sampler \citep{del2006sequential},
alternating importance sampling and Gibbs moves. 
This strategy will be employed in Section \ref{subsec:2by2}.

\section{Applications \label{sec:application}}

We present two applications. In both examples, the $(\pp, \qq, \rr)$
probabilities require distributional information about the entire random
polytopes, and not only the extreme vertices elicited in
\cite{dempster1972class}.  Both examples involve $K=4$ categories and curves in
the simplex shown in Figure \ref{fig:surfaces}.
Our main objective is to illustrate the output of the algorithm. We briefly
recall from Section~\ref{sec:categorical} that the inferred $\pp$ and $\qq$
probabilities can be understood as the degree of evidential support ``for'' or
``against'' the hypothesis of interest based on available observations and the
model specification. The $\rr$ probability, which is a distinctive feature of
DS compared to standard Bayes, indicates a degree of epistemological
indeterminacy, with a larger value encouraging the analyst to suspend judgment
about the assertion of interest. The $\rr$ probability can be useful to make
decisions or to postpone them, as with other types of imprecise probabilities
and robust Bayesian analysis \citep{berger1994overview}.  The way
that decisions can be informed by DS uncertainties has received some attention, for example
see Section 12 of \citet{shafer1990perspectives}, and also
\citet{yager1992decision,bauer1997approximation}.

\begin{figure}[t]
    \centering
    \subfloat[\label{fig:indepsurface}]{\begin{centering}
        \includegraphics[width=0.35\textwidth]{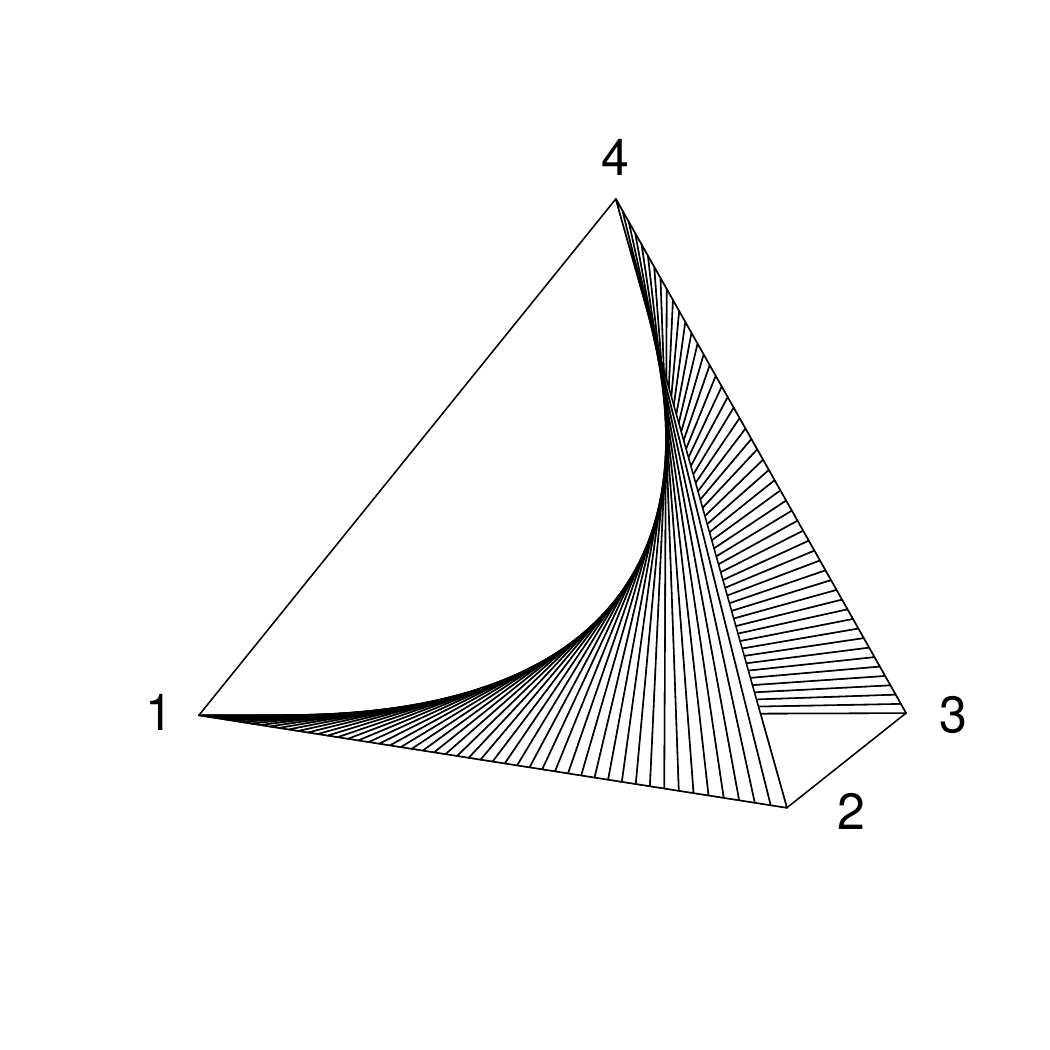}
\par\end{centering} }
\hspace*{1cm}
    \subfloat[\label{fig:linkagesurface}]{\begin{centering}
        \includegraphics[width=0.35\textwidth]{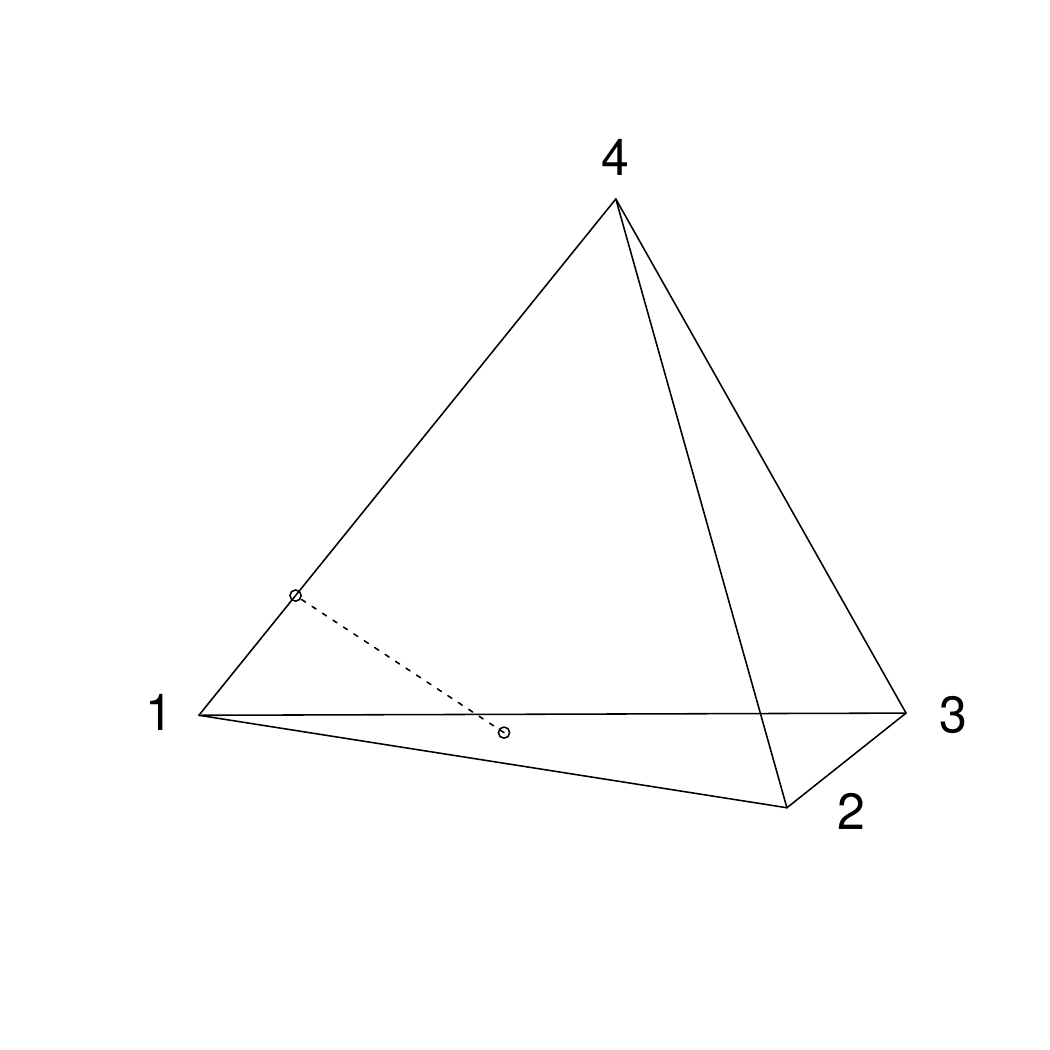}
\par\end{centering} }
\caption{ \label{fig:surfaces} Two surfaces in the 4-simplex. 
\ref{fig:indepsurface} shows the independence surface 
$\theta_1 \theta_4 = \theta_2 \theta_3$ \citep{fienberg1970geometry}. 
\ref{fig:linkagesurface} shows the linkage constraint 
of \eqref{eq:phi-constraint} for $\phi \in (0,1)$ as a dashed segment.}
\end{figure}

\subsection{Testing independence}\label{subsec:2by2}

In the case of $K=4$, count data may be arranged in a $2\times 2$
table with proportions $(\theta_{1},
\theta_{2}, \theta_{3}, \theta_{4})$ as cell probabilities, row by row.
We may be
interested in testing independence, $H_0:
\theta_{1} \theta_{4} = \theta_{2} \theta_{3}$, see \citet[Chapter 15,][]{wasserman2013all}.
Classic tests include the Pearson's chi-squared
test with $\chi^2 = \sum_{i,j}{(x_{ij} - e_{ij})^2}/{e_{ij}}$, where
$e_{ij}$ is the expected number of counts in cell
``$ij$'' under $H_0$. The Pearson test statistic is asymptotically
$\chi^2_1$. The likelihood ratio test with 
statistic $G^2 =  2\sum_{i,j} x_{ij} \log(x_{ij}/e_{ij})$, is asymptotically
equivalent; see \citet{diaconis1985testing} for further interpretations. 

Evaluating the posterior probability of $H_0$ 
raises the issue that the set $\{\theta\in\Delta\colon \theta_{1}
\theta_{4} = \theta_{2} \theta_{3}\}$, a surface in the 4-simplex as depicted
in Figure~\ref{fig:indepsurface}, might be of zero measure under the
posterior.  
As a remedy one can employ Bayes factors
\citep[e.g.][]{albert1983bayesian},
or we can consider the
evidence towards either positive or negative association,
i.e. $H_+: \theta_{1} \theta_{4} \geq \theta_{2} \theta_{3}$ or $H_-:
\theta_{1} \theta_{4} \leq \theta_{2} \theta_{3}$, 
and interpret such evidence 
as being against independence.

We consider the data set presented in 
\citet[p.191]{rosenbaum2002observational} 
regarding the effect of drainage pits on incident survival in the London
underground. Some stations are equipped with drainage pits below the
tracks. Passengers who accidentally fall off the 
platform may seek refuge in the pit to avoid an incoming train.
For stations without a pit, only 5 lived out of 21
recorded incidents. In the presence of a pit, 18 out of 32 lived.
\cite{ding2019model} reanalyzed the data to assess the difference in mortality
rates. Their analysis suggests that the
existence of a pit significantly increases the chance of survival.
The data can be summarized as counts $(16, 5, 14, 18)$.  
Pearson's chi-squared test statistic is $\chi^2 = 5.43$ with a p-value of $0.02$, 
while the likelihood ratio test yields a p-value of $0.017$.
The Bayesian analysis shows strong evidence for positive association,
with posterior probabilities
${P}(H_+ \mid {\bf x}) = 0.99$
and ${P}(H_- \mid {\bf x}) = 0.01$.

The DS approach applied sequentially yields the results 
shown in Figure \ref{fig:sequential}.
The horizontal axis shows the observations, in an arbitrary order.
The dark ribbon 
tracks $\pp(H_+)$ and $(1-\qq(H_+))$ by its lower and upper rims, respectively.
The ``don't know'' probability $\rr(H_+)$,
represented by the width of the ribbon,
can be seen to progressively shrink,
but not systematically.
The support for $H_+$ increases with each
observation in $\{1,4\}$ and decreases with each observation in $\{2,3\}$
(as highlighted with background shades).
Figure \ref{fig:sequential} is inspired by Figure 4 of \citet{walley1996analysis}.
In DS inference, the width of the ribbon is part of the inference and could be
used, for example, to inform decisions about the collection of additional
data.

\begin{figure}[t]
    \centering
        \includegraphics[width=0.9\textwidth]{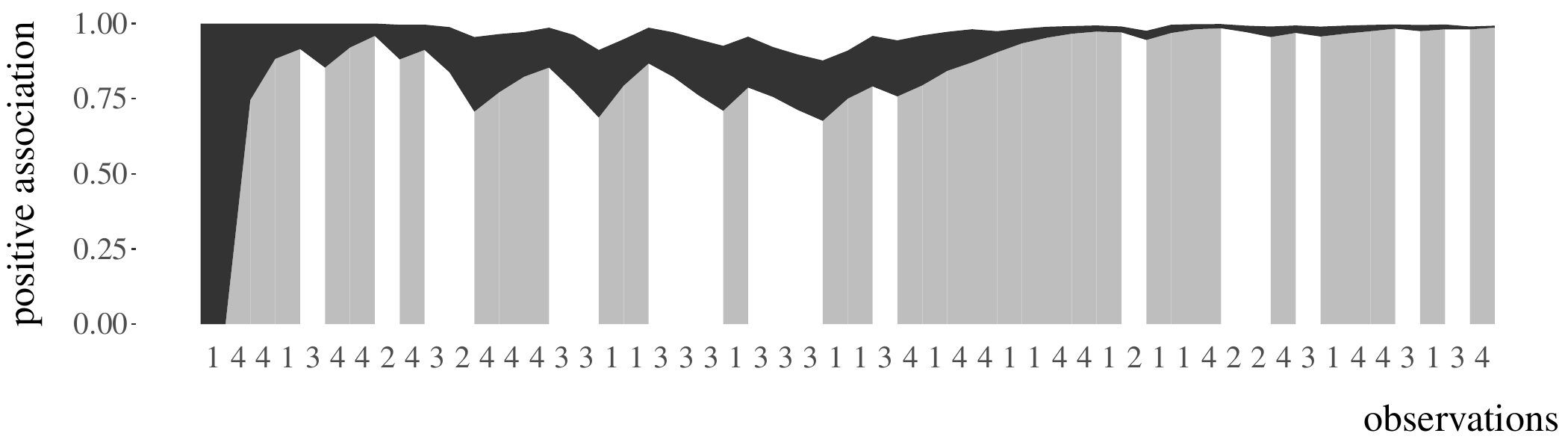}
\caption{ \label{fig:sequential} Support for the hypothesis of positive association
$H_+: \theta_1 \theta_4 \geq \theta_2 \theta_3$ as observations in $\{1,2,3,4\}$ 
are incorporated one by one. The dark ribbon delineates the probability $\pp$ for $H_+$, and one minus the support against it, respectively as its lower and upper rims. The width of the ribbon represents the amount of ``don't know'' about the hypothesis. }
\end{figure}

\subsection{Linkage model \label{subsec:linkage}}

The linkage model from  \citet[pp.368-369]{rao1973linear} was considered by
\cite{lawrence2009new}, as an example illustrating inference 
with an additional constraint. 
They compare the Imprecise Dirichlet Model (IDM) of
\cite{walley1996inferences} and their method termed Dirichlet DSM
(for Dempster--Shafer Model). The
data consist of $N = 197$ counts over $K = 4$ categories,
with probabilities satisfying
\begin{equation}
	{\theta}(\phi) = \left( \frac{1}{2} + \frac{\phi}{4}, \frac{1-\phi}{4}, \frac{1-\phi}{4}, \frac{\phi}{4} \right),\label{eq:phi-constraint}	
\end{equation}
for some $\phi \in (0,1)$. In other words, 
$\theta(\phi) = A\phi+b$ for appropriately defined $4\times 1$ matrices $A$ and $b$,
as shown in 
Figure \ref{fig:linkagesurface}.
The original observations were $(125,
18, 20, 34)$, but \cite{lawrence2009new} considered the counts 
$(25, 3, 4, 7)$, which results in a more visible 
amount of ``don't know'' probability. 

We briefly introduce Dirichlet DSM and 
focus on the comparison between the approaches.  They
differ by the choice of sampling mechanism: instead of using the
mechanism described in \citet{Dempster66}, 
\citet{lawrence2009new} introduced another mechanism 
in order to make inference simpler
computationally.
For a vector of counts $\left(N_1,...,N_K\right)$, the Dirichlet
DSM model expresses its posterior inference for the proportion vector
${\theta}$ via the random feasible set
$
\{ {\theta} \in \Delta: \theta_1 \ge z_1, ..., \theta_K \ge z_K\},
$
where ${\bf z} = \left(z_0, z_1, ..., z_K\right) \sim \text{Dirichlet}_{K+1}(1,
N_1,...,N_K)$. Incorporating the parameter constraint $\theta = A\phi+b$, the
feasible set for $\phi$ is $[\phi_{\min}({\bf z}),\phi_{\max}({\bf z})]$
with 
\begin{equation}
\phi_{\min}({\bf z}) \equiv \max\left(4z_1 - 2, 4z_4\right) \le \phi_{\max}({\bf z}) \equiv \min\left(1 - 4z_2, 1 - 4z_3\right). \label{eq:idm-phi-constraint}	
\end{equation}

For the approach of \citet{Dempster66}, termed ``Simplex-DSM'' in \citet{lawrence2009new}, 
we first run the proposed Gibbs sampler without taking into
account the linear constraint \eqref{eq:phi-constraint}. 
Among the generated feasible sets,
only those that intersect with the linear constraint are retained, and an
interval $[\underline{\phi},  \overline{\phi}]$ is obtained for each such set,
where
\[
\underline{\phi}=\text{argmin}_{\phi}\left\{ \theta\left(\phi\right)\in\mathcal{F}\left({\bf u}\right)\right\}, \qquad
 \overline{\phi}=\text{argmax}_{\phi}\left\{ \theta\left(\phi\right)\in\mathcal{F}\left({\bf u}\right)\right\}. 
\]
For the data considered here, this retains $5\%$ of the
iterations, and is therefore a practical solution. However
the approach would become impractical if the counts were much less
``compatible'' with the linkage constraint, in which case novel computational methods would be necessary.
We estimate $(\pp,\qq,\rr)$ for sets $\{\phi\in [0,c)\}$ for
$c\in(0,1)$, i.e. lower and upper cumulative
distribution functions, under both approaches and represent them 
in Figure \ref{fig:linkage:cdf}. The plot shows the
overall agreement between the two approaches. Figure \ref{fig:linkage:rcdf}
highlights the difference in $\rr$ values, and illustrates
that multiple approaches within the DS framework lead to different results.

\begin{figure}[t]
    \centering
    \subfloat[\label{fig:linkage:cdf}]{\begin{centering}
        \includegraphics[width=0.35\textwidth]{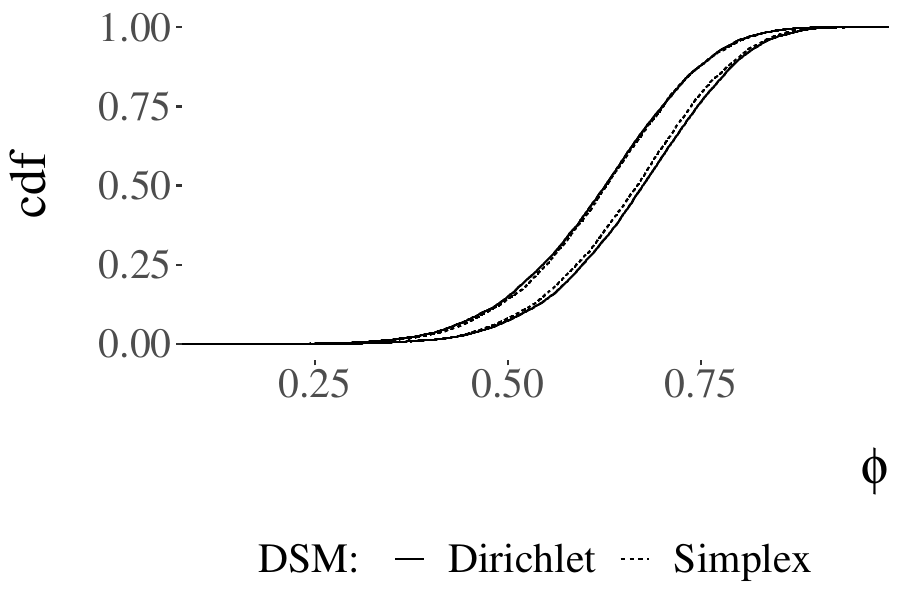}
\par\end{centering} }
\hspace*{1cm}
    \subfloat[\label{fig:linkage:rcdf}]{\begin{centering}
        \includegraphics[width=0.35\textwidth]{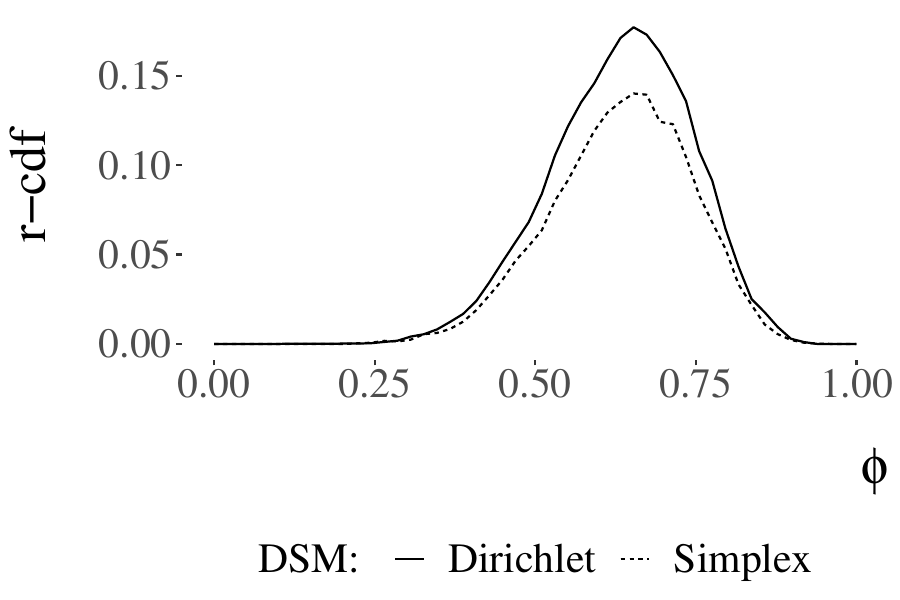}
\par\end{centering} }
\caption{\label{fig:linkage} ``Dirichlet-DSM'' approach of \citet{lawrence2009new} and the original approach of \citet{Dempster66}, for the linkage model with data $(25,3,4,7)$, the latter implemented with the proposed Gibbs sampler. Here \ref{fig:linkage:cdf} shows the lower and upper probabilities for  assertions of the form $\{\phi < c\}$ for increasing $ c \in [0, 1]$, while \ref{fig:linkage:rcdf} depicts  the difference between these upper and lower probabilites, or equivalently the $\rr$ values.}
\end{figure}

\section{Discussion\label{sec:discussion}}

The discipline of statistics does not have a single framework
for parameter inference. The setting of count data is rich enough to contrast various 
approaches.  Before any other considerations,
for a framework to be useful to scientists and decision-makers,
the ability to perform the associated computation is essential, and
allows for grounded discussions and concrete comparisons. 
Our work helps with the computation in the DS framework
for Categorical distributions,
which will hopefully motivate further theoretical
investigations of its statistical features.

One of the appeals of the DS framework is its flexibility to incorporate types
of partial information which are difficult to express in the Bayesian
framework. This includes vacuous or partial priors, coarse data which arise
from imprecise measurement devices and imperfect surveys.  These elements can
be represented as random sets
\citep{nguyen2006introduction,plass2015statistical} in the DS framework while
circumventing assumptions about the coarsening mechanism
\citep[e.g.][]{heitjan1991ignorability}.

Whether a perfect sampler could be devised 
as an alternative to the proposed Gibbs sampler is an open question.
Generic algorithms for uniform sampling
on polytopes \citep{vempala2005geometric,narayanan2016randomized,chen2018fast} 
could also provide competitive results.
The proposed Gibbs sampler could itself be accelerated, for instance by
using warm starts in the linear program solvers over subsequent iterations.

The typical challenge of DS computations is the generation of non-empty
intersections of random sets. The proposed Gibbs sampler can be seen as a way
of avoiding inefficient rejection samplers in the setting of inference in
Categorical distributions.  It remains to see whether similar ideas can be used
to deploy the DS framework in other models, for example to avoid rejection
sampling in the linkage model, or in 
hidden Markov models and models with moment constraints \citep{chamberlain2003nonparametric,bornn2019moment},
which are natural extensions of the Categorical distribution.

\paragraph*{Acknowledgements}
The authors thank Rahul Mazumder for useful advice on linear programming. 
The plots were generated using \texttt{ggplot2} \citep{ggplot2citation}
and \texttt{R} \citep{Rsoftware}.
The authors gratefully acknowledge support from the National Science Foundation
(DMS-1712872, DMS-1844695, DMS-1916002), and the National Institute of Allergy
and Infectious Disease at the National Institutes of Health [2 R37 AI054165-11
and 75N93019C00070]. The content is solely the responsibility of the authors
and does not necessarily represent the official views of the National
Institutes of Health.

\bibliographystyle{abbrvnat}
\bibliography{fid}

\appendix
\section{The choice of sampling mechanism}\label{appx:sampling-mechanism}

In DS inference, the sampling mechanism of the observable data is defined by a
structural equation which deterministically associates the data $x$ with the
parameter $\theta$ and an auxiliary random variable $u$, of an appropriate
dimension and a known distribution. Inference for $\theta$ is obtained by
inverting the structural equation and fixing $x$ at its observed value,
resulting in a map from $u$ to subsets of the parameter space.  In the present
work, we have denoted this map as $\mathcal{F}$, suppressing its dependence on
$x$.

Different sampling mechanisms may result in different
DS inference for the parameter $\theta$, even if both mechanisms correspond to 
the same likelihood of the data $x$ as a function of $\theta$.
We give an example in the case $K=2$, and consider two sampling
mechanisms corresponding to $x \sim
\text{Binomial}(N,\theta)$, where $\theta\in[0,1]$. 
The first sampling mechanism is essentially the same
as that described in the main text,
but adapted to the case where $K=2$. 
We define $u$ as an $N$-vector $(u_1,\ldots,u_N)$
of independent Uniform$(0,1)$ variables,
and the structural equation
\begin{equation}\label{eq:binom1}
  x = \sum_{n=1}^{N}{\mathds{1}}\left( u_{n}\leq \theta\right).
\end{equation}
Upon observing $x$, inverting the equation gives a map from $(u_1,\ldots,u_N)\in[0,1]^N$ to 
measurable subsets of $[0,1]$. The random sets
are of the form 
$[u_{\left(x\right)},u_{\left(x+1\right)}]$,
random closed intervals with end points given by the $x$-th and $x+1$-th
order statistics of $u_1,\ldots,u_N$, 
with the convention $u_{(0)}=0, u_{(N+1)}=1$.

The second sampling mechanism 
amounts to an inverse transform method applied to the Binomial distribution,
and involves
the quantile function $Q_{a,b}$ of the $\text{Beta}(a,b)$ distribution.
Consider the equation
\begin{equation}\label{eq:binom2}
  x=\sum_{n=0}^{N-1}\mathds{1}\left(Q_{n+1,N-n}\left(u\right) \leq \theta\right),
\end{equation}
where $u\sim \text{Uniform}(0,1)$. 
Using the fact that $(x,u)\mapsto Q_{x+1,N-x}(u)$
is increasing in $x$ for fixed $u$,
it can be shown that $x$ follows a
$\text{Binomial}(N,\theta)$.
For fixed $x$,
the random sets of interest are of the form 
\begin{equation}\label{eq:binom-inverse2}
  \left[Q_{x,N-x+1}\left(u\right),Q_{x+1,N-x}\left(u\right)\right],
\end{equation}
and we observe that the endpoints are distributed marginally
as the endpoints of 
$[u_{\left(x\right)},u_{\left(x+1\right)}]$,
$\text{Beta}(x,N-x+1)$
and $\text{Beta}(x+1,N-x)$ respectively. The random interval of the form \eqref{eq:binom-inverse2} is employed for Binomial inference by the Inferential Model framework \citep[Section 9.3.4]{martin2015inferential}.
We observe from \eqref{eq:binom-inverse2}
that the length of these intervals is deterministic given either endpoint,
which is not the case for the intervals of the form $[u_{\left(x\right)},u_{\left(x+1\right)}]$.
Therefore the $(\pp,\qq,\rr)$ triples associated with assertions of interest
on the parameters are in general different.

\section{Sampling mechanism and Gumbel-max trick} \label{appx:sampling-gumbelmax}

We observe an equivalence between the mechanism to sample Categorical
random variables from \citet{Dempster66} 
and the ``Gumbel-max'' trick \citep{maddison2014sampling,maddison2016concrete,JangGuPoole,paulus2020gradient}.

As in the main text, we can sample from a Categorical
distribution with parameters $\theta_1,\ldots,\theta_K$ as follows:
sample $w$ uniformly in the simplex $\Delta$, and output the integer $k\in[K]$
such that $w\in \Delta_k(\theta)$. To sample $w$ uniformly in $\Delta$,
we can sample $\tilde{w}_1,\ldots,\tilde{w}_K$ as independent Exponential(1),
and normalize: the $\ell$-th component $w_\ell$ of $w$ is $\tilde{w}_\ell / \sum_{j\in[K]} \tilde{w}_j$.
According to Lemma 5.2 in \citet{Dempster66},
the output $k\in[K]$ is such that 
$w_\ell/w_k \geq \theta_\ell / \theta_k$ for all $\ell\in [K]$.
Taking the logarithm and rearranging, this is equivalent to
\[\forall \ell \in [K]\quad \log \theta_\ell - \log \tilde{w}_\ell \leq \log \theta_k - \log \tilde{w}_k.\]
Since each $\tilde{w}_\ell$ is Exponential(1), $- \log \tilde{w}_\ell$ is a standard Gumbel.
Therefore the procedure is equivalent to selecting
\[k = \text{argmax}_{\ell \in[K]} \log \theta_\ell + G_\ell,\]
where $(G_\ell)$ are independent Gumbel. This is exactly the ``Gumbel-max trick'', see 
equation (2) in \citet{maddison2014sampling}.
An appeal of that sampling mechanism 
is that the categories ``are to be treated without regard to order'' \citep{Dempster66},
whereas the more common ways of sampling from a Categorical distribution
involve ordering the categories.
A similar appeal of the Gumbel-max trick 
is described in \citet{oberst2019counterfactual}.

\section{Convergence rate when $K=2$} \label{appx:cvgrate}

We consider the convergence of the proposed Gibbs sampler
in the case of two categories, $K=2$.
The analysis quantifies the impact of the counts $(N_1,N_2)$
on the convergence of the chain, which aligns with numerical
experiments shown in the main text. In the case $K=2$
the Gibbs sampler is of course unnecessary, as we can sample 
the feasible sets simply by sorting Uniform draws, as described in Appendix \ref{appx:sampling-mechanism}.

In the case $K=2$, $\Delta$ is the interval $[0,1]$ and the Gibbs sampler takes a particularly
explicit form. We assume $N_1\geq 1,N_2\geq 1$. The conditional distributions of $u_{n},n\in\mathcal{I}_{k}$
given all $u_{n},n\in\mathcal{I}_{j}$ under $\nu_{\mathbf{x}}$ are
\begin{align*}
\nu_{\mathbf{x}}((u_{n})_{n\in\mathcal{I}_{1}}|(u_{n})_{n\in\mathcal{I}_{2}}) & \propto\prod_{n\in\mathcal{I}_{1}}\mathds{1}(u_{n}\in(0,\min_{m\in\mathcal{I}_{2}}u_{m})),\\
\nu_{\mathbf{x}}((u_{n})_{n\in\mathcal{I}_{2}}|(u_{n})_{n\in\mathcal{I}_{1}}) & \propto\prod_{n\in\mathcal{I}_{2}}\mathds{1}(u_{n}\in(\max_{m\in\mathcal{I}_{1}}u_{m},1)),
\end{align*}
so that both conditionals are products of Uniform distributions.
The variables $\max_{n\in\mathcal{I}_{1}}u_{n}$
and $\min_{n\in\mathcal{I}_{2}}u_{n}$ 
can be sampled from directly from translated and scaled Beta distributions.
The Gibbs sampler alternatively updates these
two variables, which we denote by $Y$ and $Z$ below. The
procedure is described in Algorithm \ref{alg:firstkind}. Note that
$Y^{(t)}\leq Z^{(t)}$ almost surely for all $t\geq 0$ and
the feasible sets ``$\mathcal{F}(\mathbf{u})$''  are the intervals $[Y^{(t)},Z^{(t)}]$. 

\begin{algorithm}
\begin{itemize}
\item Initialization: e.g. set $Y^{(0)}=1/2$, $Z^{(0)}=1/2$.
\item At iteration $t\geq1$,
\begin{itemize}
\item Draw $Y^{(t)}=Z^{(t-1)}\times\text{Beta}(N_{1},1)$,
\item Draw $Z^{(t)}=Y^{(t)}+(1-Y^{(t)})\times\text{Beta}(1,N_{2})$.
\end{itemize}
\item Output chain $(Y^{(t)},Z^{(t)})_{t\geq0}$.
\end{itemize}
\caption{\label{alg:firstkind}Gibbs sampler for the Bernoulli model, $K=2$.}
\end{algorithm}

We can re-write the evolution of $Z^{(t)}$ as 
\[Z^{(t)} = B_1(1-B_2)Z^{(t-1)} + B_2,\]
where $B_1\sim \text{Beta}(N_1,1)$
and $B_2\sim \text{Beta}(1,N_2)$ are independent. Then we recognize
that the chain $(Z^{(t)})$
is an auto-regressive process with random coefficients.
We can analyze its rate of convergence with a coupling strategy.
Introduce another chain $(\tilde{Z}^{(t)})$, started at stationarity
and evolving with the same random input as $(Z^{(t)})$.
Then we can write 
\[
  \mathbb{E}\left[\left|Z^{(t)}-\tilde{Z}^{(t)}\right| \mid Z^{(t-1)},\tilde{Z}^{(t-1)}\right]=\mathbb{E}\left[B_{1}(1-B_{2})\right]\cdot |Z^{(t-1)}-\tilde{Z}^{(t-1)}|.
\]
Since $B_1$ and $B_2$ are independent Beta distributions, by induction we can compute
\[
  \mathbb{E}\left[\left|Z^{(t)}-\tilde{Z}^{(t)}\right|\right]\leq 
  \left(\frac{N_1}{N_{1}+1}
  \times \frac{N_2}{N_{2}+1}\right)^{t} \cdot \mathbb{E}\left[\left|Z^{(0)}-\tilde{Z}^{(0)}\right|\right].
\]
Therefore the chains contract towards one another on average, at a rate that depends on $N_1,N_2$.
This leads to geometric convergence in the Wasserstein distance \citep[e.g Theorem 2.1 in][]{gibbs2004}.
Using manipulations, the inequality $\mathbb{E}[|Z^{(0)}-\tilde{Z}^{(0)}|]\leq 1$ and
the inequality $\log(x)\leq x-1$ for all $x>0$, the Markov chain $(Z^{(t)})$
is less than $\epsilon$ from stationarity in the 1-Wasserstein distance,
with $\epsilon\in(0,1)$, for all $t$ larger than
\[\frac{(N_1+1)(N_2+1)}{N_1+N_2+2}(-\log(\epsilon)).\]
The above is an upper bound on the $\epsilon$-mixing time of the chain.
If both $N_1$ and $N_2$ increase proportionally to the sum $N=N_1+N_2$,
the bound increases linearly in $N$, which is in agreement
with figures in the main text. If $N_1$  is fixed while $N_2$ increases
to infinity, the mixing time is upper bounded by a constant independent of $N$.

\section{Estimated convergence rate based on coupled chains} \label{appx:empiricalcvgrate}

We provide a quick description of the method of \citet{biswas2019estimating}
to obtain upper bounds on the distance between a chain and its limiting distribution,
that can be estimated by Monte Carlo simulations. We consider the total variation distance,
defined for two variables $X$ and $Y$ as
\[\|\mathcal{L}(X) - \mathcal{L}(Y)\|_{\text{TV}} = \frac{1}{2} \sup_{h:|h|\leq 1} |\mathbb{E}[h(X) - h(Y)]|.\]
The supremum is over functions $h$ bounded by one.
Suppose that we can construct two chains $X^{(t)}$ and $Y^{(t)}$,
evolving marginally according to a Markov kernel $P$, converging to a distribution of interest $\pi$, 
and such that $X^{(t)}=Y^{(t)}$ in distribution for all
$t\geq 0$, while $X^{(t)}=Y^{(t-L)}$ almost surely for $t\geq \tau^{(L)}=\inf\{t: X^{(t)}=Y^{(t-L)}\}$
and where $L\geq 1$ is a user-chosen ``lag'' integer.

Then \citet{jacob2017unbiased,biswas2019estimating} provide a set of assumptions under which,
for any function $h$ bounded by one, 
\begin{align*}
  \int h(x) \pi(dx) - \mathbb{E}[h(X^{(t)})]  = \mathbb{E}\left[\sum_{j=1}^{\lceil (\tau^{(L)}-L-t)/L\rceil} \left(h(X^{(t+jL)}) - h(Y^{(t+(j-1)L)})\right)\right].
\end{align*}
This comes from a telescopic sum argument first derived in the context of coupled Markov chains by \citet{glynn2014exact}.
From there, conditioning on $\tau^{(L)}$, taking triangle inequalities and upper bounding
$|h(x)-h(y)|$ by $2$ for all $x,y$, we obtain 
\begin{align*}
  \|\mathcal{L}(X^{(t)}) - \pi \|_{\text{TV}}&\leq \mathbb{E}\left[\max(0, \lceil (\tau^{(L)}-L-t)/L\rceil)\right].
\end{align*}
The right-hand side can be estimated by generating independent copies of $\tau^{(L)}$,
in other words by sampling coupled Markov chains with a lag and recording their meeting time $\tau^{(L)}$,
and then by approximating the expectation by an empirical average. 

It remains to describe the coupling employed for the proposed 
Gibbs sampler. 
Denoting by $P$ the transition kernel of the Gibbs sampler, so that 
$\mathbf{u}^{\prime} \sim P(\mathbf{u},\cdot)$
given $\mathbf{u}\in\Delta^{N}$, 
we construct a coupled transition kernel
$\bar{P}$ on $\Delta^{N}\times\Delta^{N}$. The second chain is denoted by 
$\tilde{\mathbf{u}}$.
We employ the following coupling strategy. 
With probability $\omega \in (0,1)$, we propagate the two chains using
common random numbers. As we have established in the case $K=2$
in Section \ref{appx:cvgrate}, this induces a contraction between the chains.
With the remaining probability, conditional updates are maximally
coupled \citep[e.g][]{thorisson2000coupling,jacob2017unbiased}, so that we have a chance
to observe $\mathbf{u}^{(t)}=\tilde{\mathbf{u}}^{(t-L)}$.
This mixture strategy is 
similar to that employed in \citet{heng2019unbiased} for Hamiltonian Monte Carlo samplers.
The mixing parameter $\omega$ was set to $0.9$ throughout our
experiments. 

For the choice of lag, we have used $L=50$ for Figure 5(a) in the main text
and $L=10\times N$ for Figure 5(b), and both figures
are obtained from $500$ independent copies of $\tau^{(L)}$,
which took a few minutes to run on a personal computer from 2019
with a 2.4 GHz Intel Core i9, using 7 cores, 14 threads.

\end{document}